\documentclass[12pts]{article}

\RequirePackage[OT1]{fontenc}
\RequirePackage[amsthm,amsmath,natbib,noinfoline]{imsart}

\usepackage{latexsym}
\usepackage{amsfonts}
\usepackage{amssymb}
\usepackage{fontenc}
\usepackage{textcomp}
\usepackage{graphics}
\usepackage{color}
\usepackage{colortbl}
\usepackage{graphicx}

\startlocaldefs

    \theoremstyle{plain}
    \numberwithin{equation}{section}
    \newtheorem{thrm}{Theorem}[section]
    \newtheorem{prop}[thrm]{Proposition}

    \newtheorem{lmma}[thrm]{Lemma}
    \newtheorem{remk}[thrm]{Remark}
    \newtheorem{defn}[thrm]{Definition}
    \newtheorem{exmpl}[thrm]{Example}

    \def\bbr{{\mathbb R}}
    \def\bbe{{\mathbb E}}
    \def\bbf{{\mathbb F}}
    \def\bbp{{\mathbb P}}
    \def\bbq{{\mathbb Q}}
    
    \def\bbh{{\mathbb H}}

    \def\bbl{{\mathbb L}}
    
    \def\calB{{\mathcal B}}

    \def\calE{{\mathcal E}}
    \def\calF{{\mathcal F}}
    
    \def\calM{{\mathcal M}}

    \def\calA{{\mathcal A}}

    \def\calU{{\mathcal U}}
    \def\calV{{\mathcal V}}
    
    \def\calP{{\mathcal P}}
    \def\calS{{\mathcal S}}

\endlocaldefs

\begin{document}
\begin{frontmatter}

\title{State-dependent utility maximization in L\'evy markets}
\runtitle{Utility maximization in L\'evy markets}

\begin{aug}
\author{\fnms{Jos\'e} E. \snm{Figueroa-L\'opez}
\ead[label=e1]{figueroa@stat.purdue.edu}}
\and
\author{\fnms{Jin} \snm{Ma}
\ead[label=e2]{jinma@usc.edu}}



\runauthor{J.E. Figueroa-L\'opez and J. Ma}

\address{Department of Statistics\\
Purdue University\\
West Lafayette, IN 47906\\
\printead{e1}}

\address{Department of Mathematics\\
University of Southern California\\
Los Angeles, CA 90089\\ 
\printead{e2}\\
\vspace{.3 cm}
}
\end{aug}


\begin{abstract}
    We revisit Merton's portfolio optimization problem
    under boun-ded state-dependent utility functions,
    in a market driven by a L\'evy process $Z$ extending results by 
    Karatzas et. al. \cite{Karatzas:1991} and Kunita
\cite{Kunita:2003}.
    The problem is solved using a
    \emph{dual variational problem} as it is customarily done
    for non-Markovian {models}.
    One of the main features here is that the domain of the
    \emph{dual problem} enjoys an explicit 
    ``parametrization'', built on
    a \emph{multiplicative optional decomposition}
    for nonnegative supermartingales
    due to F\"ollmer and Kramkov \cite{Follmer:1997}.
    As a key step in obtaining the representation result we prove a \emph{closure property}
    for integrals with respect to Poisson random
    measures, {a result of interest on its own that extends the analog property for integrals with
respect to a fixed semimartingale due to M\'emin \cite{Memin}}. 
    {In the case that (i) the L\'evy measure $\nu$ of $Z$ is atomic with a finite number of atoms
    or that (ii) $\Delta S_{t}/S_{t^{-}}=\zeta_{t} \vartheta(\Delta Z_{t})$ for a process $\zeta$ and a deterministic function $\vartheta$, we explicitly characterize the admissible trading strategies and show that the dual solution is a risk-neutral local martingale.}
\end{abstract}



\end{frontmatter}

\section{Introduction}
The task of determining good trading strategies is a fundamental
problem in mathematical finance. A typical approach to this
problem aims at finding the trading strategy that maximizes, for
example, the final expected utility, which is defined as a
deterministic, concave, and increasing function
$U:\bbr\rightarrow\bbr\cup \{-\infty\}$ of the final wealth. There
are, however, many applications where a utility function has to
change with the underlying securities, or more generally, with the
source of randomness (say a Brownian motion). For example, in the
so-called {\it optimal partial replication} of a contingent claim,
introduced by F\"ollmer and Leukert \cite{Follmer:2000}, one tries
to find the trading strategy that best replicates the claim $H$
under a budget constraint.
In particular, when the market is incomplete, it is often more
beneficial to allow certain degree of ``shortfall" in order to
reduce the ``super hedging cost", a threshold for the minimum
initial wealth so that super-hedging is feasible (see, e.g.,
\cite{DlbSch:1994} and \cite{Kramkov:1996} for more details).
Mathematically, such a shortfall risk could be measured
by the expected loss
\[
    \bbe\left[L\left((H-V_{_{T}})^{+}\right)\right],
\]
where $L$ is the ``loss function",
a convex increasing function that incorporates
the investor's attitude towards the shortfall
$(H-V_{_{T}})^{+}$, and the value process $V$ is subject to
the constraint $V_{0}\leq{}z$. Such a problem can then be
formulated as
a utility maximization problem with
a bounded \emph{state-dependent utility}, in which the utility
function is defined by (cf. \cite{Follmer:2000}):
 \begin{equation}
 \label{U}
 U(v;\omega):=L(H(\omega))-L((H(\omega)-v)^{+}), \qquad \omega\in\Omega.
 \end{equation}

In general, we can define a \emph{state-dependent utility} as a
function  $U:\bbr_{+}\times\Omega\rightarrow\bbr_{+}$ such that
$U(\cdot;\omega)$ is a utility function for each
$\omega\in\Omega$.
The utility maximization problem is then defined as
\begin{equation}\label{PrimalProblemReal}
    u(z):=
    \sup\left\{ \bbe\left[U(V_{_{T}}(\cdot),\cdot)\right]
        :\; V\;{\rm is}\;{\rm admissible}\;{\rm and }\;
        V_{0}\leq{}z \right\},
    \end{equation}
where the supreme is taken over all wealth processes
$\{V_{t}\}_{t\leq{}T}$ generated by admissible trading strategies
(see Section \ref{FinModelSect} for a precise definition).

The existence and essential uniqueness of the solution to the
problem (\ref{PrimalProblemReal}) was proved in
\cite{Follmer:2000} for a general semimartingale price model using
a convex duality method, built on a celebrated bipolar theorem by
Kramkov and Schachermayer \cite{Kramkov:1999}. However, this
approach does not seem to shed
any light on how to compute, in a feasible manner, the optimal
trading strategy. This is partly due to the generality of the
problem considered there. In this paper we shall consider
the market model in which the price is driven by a
L\'evy process, and we propose a more manageable dual problem with
a specific domain.
We should note that our method can be extended to handle more
general jump-diffusion models driven by even additive
processes.

The problem of utility maximization can be traced back to
Merton \cite{Merton:1969}-\cite{Merton:1971}.
{In a Brownian-driven market model,  Karatzas et. al. \cite{Karatzas:1991} 
developed a program, known 
as the \emph{convex duality method}, that has become one of the most powerful methods, yet relatively explicit and simple, to analyze optimal portfolio problems in non-Markovian markets. 
They prove that} the marginal utility of the optimum final wealth
is proportional to the risk-neutral local martingale
that minimizes a ``dual" problem, defined as another optimization
problem with the objective function being the
Legendre-Fenchel-type transformation of the original utility
function. 
To be more precise,
consider a minimization problem
 \begin{equation}
 \label{DualProb0}
     v_{_{\Gamma}}(y)=\inf_{\xi\in\Gamma} \bbe\left[\widetilde{U}\left(y
     \xi_{_{T}}\right)\right], \qquad y>0,
 \end{equation}
where  $\Gamma$,
the so-called \emph{dual domain or class},
consists of (risk-neutral) exponential local martingales,
and $\widetilde{U}(\cdot)$ stands for the
\emph{convex dual function} of $U(\cdot)$.
The idea is first to find, for any $y>0$, a minimizer
$\xi^*_y\in\Gamma$ of (\ref{DualProb0}),
which {in turn} induces
a {``potential"} optimal terminal wealth $V^*_y$ {in the sense that the so-called weak duality relation 
\begin{equation}\label{WD} 
	u(z)\leq \bbe \left[ U(V^{*}_{y})\right]
\end{equation}
holds.} 
If one can further show that for some $y^*>0$,
there exists an admissible portfolio $\beta^*$ such that
$V^{\beta^*}_{_{T}}\geq {V}^{*}_{y^{*}}$,
then clearly {equality holds in (\ref{WD}) (a property typically called \emph{strong duality})}, 
and $\beta^*$ solves the original problem.
Customarily, finding the optimal portfolio $\beta^*$
relies on a variational problem for the dual value
function and the existence of the minimizer $\xi^*_{y^*}$
utilizes the particular form of the market model
 and some general properties of the utility function.

{More recently}, the convex duality method {was further} extended to
a general ``jump-diffusion'' market by Kunita
\cite{Kunita:2003} 
{building on} an exponential representation for \emph{positive} local
supermartingales as well as
a variational equality for the dual problem.
To ensure the attainability of the dual problem,
it is required that
the utility function satisfies the same conditions as
\cite{Karatzas:1991} {(one of which is unboundedness)},
and that the dual domain $\Gamma$ contain all positive
``risk-neutral'' local supermartingales.

The main purpose of this paper is to further extend the {seminal} approach of
\cite{Karatzas:1991} to the case of
\emph{state-dependent, bounded} utility functions. For simplicity,
we will be contented with a market with only one stock, whose jumps
are driven by a L\'evy process {$Z:=\{Z_{t}\}_{t\geq{}0}$}, but our analysis can be readily
extended to more general jump-diffusion multidimensional models such
as the one considered in \cite{Kunita:2003}. We should emphasize
that the boundedness and potential non-differentiability of the
utility function causes some technical subtleties. For example, the
dual optimal process can be $0$ with positive probability, thus the
representation theorem of Kunita does not apply anymore. To get
around these difficulties we shall reconsider the dual problem over
an arbitrary subclass. Using an exponential representation for
nonnegative supermartingales due to F\"ollmer and Kramkov
\cite{Follmer:1997}, we show how to construct suitable explicit dual
classes associated with {certain} classes of semimartingales that are closed
under \'Emery's topology. To work with this last condition, we prove
a closure property for integrals with respect to Poisson random
measures, {a result of interest on its own that extends} the analog property for integrals with
respect to a fixed semimartingale due to M\'emin \cite{Memin}. It is
also worth mentioning that part of our approach relies on the
fundamental characterization of contingent claims that are
super-replicable (see \cite{DlbSch:1994} and \cite{Kramkov:1996}),
while that of  F\"ollmer and Leukert \cite{Follmer:2000} (see also
Xu \cite{Xu:2004}) was based on the bipolar theorem of
\cite{Kramkov:1999}. We feel that 
{the convex duality approach of \cite{Karatzas:1991} that we develop in this paper offers several advantages. 
The proofs are more direct and the dual problem might be more
suitable for computational purposes since the dual class enjoys an
``explicit'' description and ``parametrization''.}  
    {In the case that (i) the jumps of the price process $S$ are driven by the superposition of finitely-many shot-noise Poisson processes or that (ii) $\Delta S_{t}/S_{t^{-}}=\zeta_{t} \vartheta(\Delta Z_{t})$, we are even able to show that the dual solution is a risk-neutral local martingale.}

We would like to remark that some of our results in Section  \ref{ConvDualSect} below may look similar to
those in \cite[Chapter 3]{Xu:2004}, but there are essential
differences. For example, the model in  \cite{Xu:2004} exhibits only
finite-jump activity, while our model allows general jumps. Also,
\cite{Xu:2004} allows only downward price jumps, an assumption that
{seems to be} crucial for the approach there, which was based on the existence
of the solutions to certain stochastic differential equations (see,
e.g., \cite[Lemma 3.3, Proposition 3.4]{Xu:2004}). 
We should point out that our approach is also valid for general
additive processes, including the time-inhomogeneous cases
considered in \cite{Xu:2004}.
We present an argument in (ii) of Section \ref{FnlRmrksSect} to
justify this point.

{The} paper is organized as follows. In Section \ref{FinModelSect} we
introduce the financial model, along with some basic terminology
that will be used throughout the paper. The convex duality method is
{revised} in Section \ref{ConvDualSect}, where a potential optimal
final wealth satisfying {(\ref{WD})} is constructed. An
explicit description of a dual class for which {equality in (\ref{WD})}
holds is presented in Section \ref{SpecDualClssSect}, along with some
interesting simple characterizations of the dual optimum. In particular, {as it was mentioned earlier,} we prove that under certain conditions in the structure of the jumps,
the dual optimum is actually a local martingale {and we also provide 
an explicit characterization of the admissible trading strategies}.  In section
\ref{RplcbltySect} we show that the potential optimal final wealth
is attained by an admissible trading strategy, as the last step for
proving the existence of {an} optimal portfolio. Finally, we give some
concluding remarks in Section \ref{FnlRmrksSect}. {Some necessary
fundamental theoretical results behind our approach are collected in Appendix
\ref{CovClssSect}, such as the exponential representation for nonnegative supermartingales of F\"ollmer and Kramkov \cite{Follmer:1997} and the closure property for integrals with respect to Poisson random measures that was previously mentioned.} 

\section{Notation and problem formulation}
\label{FinModelSect}
\setcounter{equation}{0}

Throughout this paper we assume that all the randomness comes from a
complete probability space $(\Omega, \calF, \bbp)$,
on which there is defined a
L\'evy process $Z$ with L\'evy triplet
$(\sigma^{2},\nu,0)$
(see Sato \cite{Sato} for the terminology).
By the L\'evy-It\^o decomposition,
there exist a standard Brownian motion
$W$ and an independent Poisson random measure $N$ on
$\bbr_{+}\times\bbr\backslash\{0\}$ with mean measure
$\bbe\, N(dt,dz)=\nu(dz)dt$, such that
    \begin{equation}
    \label{LevyKhin}
        Z_{t}=\sigma W_{t}+
        \int_{0}^{t}\int_{|z|\leq{}1}z\,
        \widetilde{N}(ds,dz)+
        \int_{0}^{t}\int_{|z|>1}z N(ds,dz),
    \end{equation}
where $\widetilde{N}(dt,dz):=N(dt,dz)-\nu(dz)dt$.
Let $\bbf:=\left\{\calF_{t}\right\}_{t\geq{}0}$ be the natural
filtration generated by $W$ and $N$, augmented by all the null sets in
$\calF$ so that it satisfies the {\it usual conditions} (see e.g.
\cite{Protter}).

\bigskip
\noindent{\bf The  market model} ~~
\smallskip

We assume that there are two assets in the market: a risk free
bond (or money market account), and a risky asset, say, a stock.
The case of multiple stocks,
such as the one studied in \cite{Kunita:2003}, can be treated in a
similar way without substantial difficulties (see section
\ref{FnlRmrksSect} for more details). As it is customary
all the processes are taken to be discounted to the present value
so that the value $B_{t}$ of the risk-free asset
can be assumed to be identically equal to $1$.
The (discounted) price of the stock follows the stochastic
differential equation
\begin{equation}
 \label{EqForStock}
   d\, S_{t}= S_{t^{-}}
    \left\{b_{t}\, dt +
    \sigma_{t}\, dW_{t}+
    \int_{\bbr_{0}}v(t,z)\widetilde{N}(dt,dz)
    \right\},
\end{equation}
where $\bbr_{0}:=\bbr\backslash\{0\}$, $b\in L^1_{loc}$, $\sigma \in L^{2}_{loc}(W)$,
and $v\in G_{loc}(N)$ (see \cite{Shiryaev:2003}
for the terminology). More
precisely, $b$, $\sigma$, and $v$ are predictable
processes such that $v(\cdot,\cdot) >-1$ a.s.
(hence, $S_{\cdot}>0$ a.s.), and that
\[
    \int_{0}^{\cdot}|b_{t}|dt,
    \quad
    \int_{0}^{\cdot}|\sigma_{t}|^{2}dt,
    \quad{\rm and}\quad
    (\sum_{s\leq \cdot}v^{2}(s,\Delta Z_{s}))^{1/2}
\]
are locally integrable. Finally, we assume that the market is free of arbitrage so that
there exists a risk-neutral probability measures $\bbq$
such that the (discounted) process
$S_{t}$, $0\leq t \leq T$, is an
$\bbf$-local martingale under $\bbq$.
Throughout, $\calM$ will stand for the class of all equivalent risk
neutral measures $\bbq$.

\bigskip

\noindent{\bf Admissible trading strategies and the utility
maximization problem.}

\smallskip

A trading strategy is determined by a predictable locally bounded
process $\beta$ representing the \emph{proportion of total wealth
invested in the stock}. Then, the resulting wealth process is
governed by the stochastic differential equation:
\begin{equation}\label{EqDiscValProcReal}
    V_{t}
    =w+\int_{0}^{t}V_{s^{-}}\,\frac{\beta_{s}}{S_{s^{-}}}\;d S_{s},\quad
    0<t\leq{}T,
\end{equation}
where $w$ stands for the initial endowment. For future reference, we
give a precise definition of ``{\it admissible strategies}".
\begin{defn}
The process $V^{w,\beta}:=V$
solving (\ref{EqDiscValProcReal}) is called the value
process corresponding to the self-financing portfolio
with initial endowment $w$ and
trading strategy $\beta$.
We say that a value process $V^{w,\beta}$ is ``{admissible}" or
that the process $\beta$ is ``{admissible}" for $w$ if 
\(
    V_{t}^{w,\beta}\geq{}0,\;\; \forall\, t\in [0,T].
\)
\end{defn}
For a given initial endowment $w$, we denote the set all admissible
strategies for $w$ by $\calU^w_{ad}$, and the set of all admissible
value processes by $\calV^w_{ad}$. In light of the Dol\'eans-Dade
stochastic exponential of semimartingales (see e.g. Section I.4f in
\cite{Shiryaev:2003}), one can easily obtain necessary and
sufficient conditions for admissibility.
\begin{prop}\label{CndAdmissibility}
    A predictable locally bounded process $\beta$
    is admissible if and only if
    \[
        \bbp\left[\{\omega\in\Omega:
        \beta_{t} v(t,\Delta Z_{t})\geq{}-1,\quad
        \text{for a.e. } t\leq{}T\}\right]=1.
    \]
\end{prop}

To define our utility maximization problem, we begin by
introducing the {\it bounded state-dependent utility function}.
\begin{defn}\label{CondOnUtility}
A random function $U:\bbr_{+}\times\Omega\mapsto\bbr_{+}$ is
called a ``bounded and state-dependent utility function" if
\begin{itemize}
    \item[1.]
     $U(\cdot,\omega)$ is nonnegative, non-decreasing,
    and continuous on $[0,\infty)$;
    \item[2.] 
    For each fixed $w$, the mapping $\omega \mapsto U(w,\omega)$ is $\calF_{T}$-measurable;
    \item[3.] There is an $\calF_{T}$-measurable, positive random variable $H$ such that
    for all $\omega\in\Omega$, $U(\cdot,\omega)$ is {a strictly concave differentiable function}
     on $(0,H(\omega))$, and it holds that
 \begin{eqnarray}
 \label{Ubdd}
  && U(w,\omega)\equiv U(w\wedge H(\omega),\omega), \qquad
  w\in\bbr_+; \\
 \label{EUfin}
  && \bbe\left[U\left(H;\cdot\right)\right]<\infty;
 \end{eqnarray}
\end{itemize}
\end{defn}
{
Notice that the $\calF_{T}$-measurability of the random variable
$\omega\rightarrow U(V_{T}(\omega),\omega)$
is automatic because $U(w,\omega)$ is  $\calB([0,\infty))\times\calF_{T}$-measurable
in light of the above conditions 1 and 2.}
We remark that while the assumption (\ref{EUfin}) is merely
technical, the assumption (\ref{Ubdd}) is motivated by the shortfall
risk measure (\ref{U}). Our utility optimization problem is thus
defined as
 \begin{equation}
 \label{UtiMax}
    u(z):=
    \sup\left\{ \bbe\left[U(V_{T}(\cdot),\cdot)\right]
        :\; V\in \calV^w_{ad} ~~{\rm with }~~
        w\leq{}z \right\}.
    \end{equation}
for any $z>0$. We should note that the above problem is
relevant only for those initial wealths $z$ that are smaller than
$\bar{w}:=\sup_{\bbq\in\calM}\bbe_{\bbq}\left\{H\right\}$,
the super-hedging cost of $H$. Indeed, if $z\geq \bar{w}$, then
there exists an admissible trading strategy $\beta^{*}$ for $z$
such that $V_{_{T}}^{z,\beta^{*}}\geq{}H$ almost surely, and
consequently, $u(z)=\bbe\left[U(H,\cdot)\right]$ (see
\cite{DlbSch:1994} and \cite{Kramkov:1996} for this
\emph{super-hedging} result).

Our main objectives in the rest of the paper are the following: (1)
Define the dual problem and identify the relation between the value
functions of the primal and the dual problems;  (2) By suitably
defining the dual domain, prove the attainability of the associated
dual problem; and (3) Show that the potential optimum final wealth
induced by the minimizer of the dual problem can be realized by an
admissible portfolio. We shall carry out these tasks in the
remaining  sections.

\section{The duality method and the dual problems}
\label{ConvDualSect}
\setcounter{equation}{0}

In this section we introduce the {\it dual problems} corresponding
to the primal problem (\ref{PrimalProblemReal}). We begin by
defining the so-called \emph{convex dual function} of
$U(\cdot;\omega)$:
\begin{equation}\label{ConvDual}
    \widetilde{U}(y,\omega):=
    \sup_{0\leq{}z\leq{}H(\omega)}
    \left\{U\left(z,\omega\right)
    -yz\right\}.
\end{equation}
We note that the function $\widetilde{U}$ is closely related to
the Legendre-Fenchel transformation of the convex function
$-U(-z)$. It can be easily checked that
$\widetilde{U}(\cdot;\omega)$ is convex and differentiable
everywhere, for each $\omega$. Furthermore, if we denote the {\it
generalized inverse} function of {$U'(\cdot, \omega)$} by
\begin{equation}\label{Inverse}
    I(y,\omega):=
    \inf\left\{z\in(0,H(\omega))|U'(z,\omega)<y\right\},
\end{equation}
with the convention that $\inf {\emptyset}=\infty$, then it holds
that
\begin{equation}\label{DerConvexDual}
        \widetilde{U}'(y,\omega)=-\left(I(y;\omega)\wedge H\right),
        \quad \forall\; y>0,
    \end{equation}
and the function $\tilde U$ has the following representation
\begin{eqnarray}
\label{Utilde}
    \widetilde{U}(y,\omega)=U\left(I(y,\omega)\wedge
    H(\omega),\omega\right)-y\left(I(y,\omega)\wedge
    H(\omega)\right).
\end{eqnarray}

{\begin{remk}\label{Measurability2}
    We point out that the random fields defined in
    (\ref{ConvDual})-(\ref{Inverse}) are
    $\calB([0,\infty))\times\calF_{T}$-measurable.
    For instance, in the case of $\widetilde{U}$,
    we can write
    \[
        \widetilde{U}(y,\omega)=
        \sup_{z\geq{}0}\{U(z,\omega)-yz\}{\bf 1}_{\{z\leq{}H(\omega)\}},
    \]
    and we will only need to check that
    $(y,\omega)\rightarrow\{U(z,\omega)-yz\}{\bf 1}_{\{z\leq{}H(\omega)\}}$ is
    jointly measurable for each fixed $z$. This last fact follows
    because
    the random field in question is continuous in
    the spatial variable $y$ for each $\omega$, and is 
    $\calF_{T}$-measurable for each y.
    In light of (\ref{DerConvexDual}),
    it transpires that the random field
    $I(y,\omega)\wedge H(\omega)$ is jointly measurable.
    Given that the subsequent dual problems and corresponding solutions
    are given in terms of the fields $\widetilde{U}(y,\omega)$
    and $I(y,\omega)\wedge H(\omega)$
    (see Definition \ref{DefDualProblem} and Theorem \ref{MainThrm1} below),
    the measurability of several key random variable below
    is guaranteed.
\end{remk}
}

Next, we introduce the so-called ``dual class".
\begin{defn}\label{DualClass}
    Let $\widetilde\Gamma$ be the class of nonnegative
    supermartingales $\xi$ such that
    \begin{enumerate}
    \item[(i)] $0\leq\xi(0)\leq 1$, and
    \smallskip
    \item[(ii)] for each locally bounded admissible trading
    strategy $\beta$, 
    $\{\xi(t)
    V^{\beta}_{t}\}_{t\leq{}T}$ is a
    supermartingale.
    \end{enumerate}
\end{defn}

To motivate the construction of the dual problems below we note that
if $\xi\in\widetilde{\Gamma}$ and $V$ is the value process of a
self-financing admissible portfolio with initial
endowment $V_{0}\leq{}z$, then $\bbe\left[\xi(T)
\left(V_{_{T}}\wedge H\right)\right]\leq z$, and it follows that
\begin{eqnarray}
  \label{dualrel}
    \bbe\left[U\left(V_{T},\cdot\right)\right]&\leq&
    \bbe\left[U\left(V_{T}\wedge H,\cdot\right)\right]
    -y\left(\bbe\left[\xi(T)
    \left(V_{T}\wedge H\right)\right]-z\right)\nonumber\\
    &\leq&\bbe\{\sup_{0\leq{z'}\leq H(\cdot)}\left\{U\left(z',\cdot\right)
    -y\xi(T)
    z'\right\}\}+zy\\
    &=&\bbe\{\widetilde{U}(y\xi(T)
    ,\cdot)\}+zy.
    \nonumber
\end{eqnarray}
for any $y\geq{}0$.  The dual problem is defined as follows.
\begin{defn}\label{DefDualProbl}
Given a subclass $\Gamma\subset\widetilde\Gamma$,
the minimization problem
\begin{equation}
\label{DefDualProblem}
    v_{_{\Gamma}}(y):=\inf_{\xi\in\Gamma}
    \bbe\left[\widetilde{U}(y\xi(T)
    ,\omega)\right],
    \quad y>0,
\end{equation}
is called the ``dual problem induced by $\Gamma$". The class
$\Gamma$ is referred to as a dual  domain (or class) and
$v_{_{\Gamma}}(\cdot)$ is called its dual value function.
\end{defn}
Notice that, by (\ref{dualrel}) and (\ref{DefDualProblem}), we have
the following \emph{weak duality} relation between the primal and
dual value functions:
\begin{equation}\label{WeakDualityEq}
    u(z)\leq{}v_{_{\Gamma}}(y)+zy,
\end{equation}
valid for all $z,y\geq 0$. {The effectiveness of the dual problem
depends on the attainability of the lower bound in
(\ref{WeakDualityEq}) for some $y^*=y^*(z)>0$ (in which case, we say
that \emph{strong duality} holds), and the attainability of its
corresponding dual problem (\ref{DefDualProblem}).
The following {important properties will be needed for future reference}.
\begin{prop}\label{ProptyDual}
    The dual value function $v_{_{\Gamma}}$ corresponding
    to a subclass $\Gamma$ of $\widetilde\Gamma$ satisfies the
    following properties:
    \begin{enumerate}
    \item[(1)] $v_{_{\Gamma}}$ is non-increasing on $(0,\infty)$ and
    \(
        \bbe\left[U(0;\cdot)\right]\leq
        v_{_{\Gamma}}(y)\leq \bbe\left[U(H;\cdot)\right].
    \)

    \item[(2)] If
    \begin{equation}\label{CondOnClassOfSprMart1}
        0<w_{_{\Gamma}}:=\sup_{\xi\in\Gamma}\bbe\left[\xi(T)
        H\right]<\infty,
    \end{equation}
    then $v_{_{\Gamma}}$ is uniformly continuous on $(0,\infty)$, and
    \begin{equation}\label{DervDual}
        \lim_{y\downarrow{}0}\frac{\bbe\left[U(H;\cdot)\right]-v_{_{\Gamma}}(y)}{y}=
        \sup_{\xi\in\Gamma}\bbe\left[\xi(T)
        H\right].
    \end{equation}

    \item[(3)] There exists a process $\widetilde{\xi}\in\widetilde{\Gamma}$
    such that
    \(
    \bbe\left[\widetilde{U}(y\,\widetilde{\xi}(T)
    ,\cdot)\right]
    \leq v_{_{\Gamma}}(y).
    \)

    \item[(4)] If $\Gamma$ is a convex set, then
    (i) $v_{_{\Gamma}}$ is convex, and
    (ii) there exists a $\xi^{*}\in\widetilde\Gamma$
    attaining the minimum $v_{_{\Gamma}}(y)$.
    Furthermore,
        the optimum $\xi^{*}$ can be ``approximated''
        by elements of $\Gamma$ in the sense that
        there exists a sequence $\{\xi^{n}\}_{n}\subset\Gamma$
        for which $\xi^{n}(T)\rightarrow\xi^{*}(T)$, a.s.
    \end{enumerate}
\end{prop}
\begin{proof}
    For simplicity, we write $v(y)=v_{_{\Gamma}}(y)$.
    The monotonicity and range of values of $v$ are straightforward.
    To prove (2), notice that since
     $\widetilde{U}(\cdot;\omega)$ is convex, non-increasing, and
     $\widetilde{U}'(0^{+};\omega)=-H(\omega)$, we have
     \[
        \frac{\bbe\left[U(H;\cdot)\right]-
        \inf_{\xi}\bbe\left[\widetilde{U}(y\xi(T)
        )\right]}{y}\leq
        \sup_{\xi\in\Gamma}\bbe\left[\xi(T)
        H\right].
     \]
     On the other hand, by the mean value theorem,
     dominated convergence theorem, (\ref{DerConvexDual}),
     and Assumption \ref{CondOnUtility},
     \[
        \bbe\left[H\hat\xi(T)
        \right]
        \leq{}
        \liminf_{y\downarrow{}0}\frac{\bbe\left[U(H;\cdot)\right]-v(y)}{y}
        \leq
        \sup_{\xi\in\Gamma}\bbe\left[\xi(T)
        H\right],
     \]
     for every $\hat{\xi}\in\Gamma$. Then, (2)
     is evident.
     Uniform continuity is straightforward since for any $h$
     small enough it holds that
     \[
        |v_{_{\Gamma}}(y+h)-v_{_{\Gamma}}(y)|\leq{}w_{_{\Gamma}}
        |h|,
     \]

     The part (i) of (4) is well-known. Let us turn out to
     prove (3) and part (ii) in (4).
    Let $\{\xi^{n}\}_{n\geq{}1}\subset \Gamma\subset
    \widetilde{\Gamma}$ be such that
    \begin{equation}
        \lim_{n\rightarrow\infty}\bbe\left[\widetilde{U}(y\xi^{n}_{T}
        ,\omega)\right]
        =v_{_{\Gamma}}(y).
    \end{equation}
Without loss of generality, one can assume that each process
$\xi^{n}$ is constant on $[T,\infty)$.
By Lemma 5.2  in \cite{Follmer:1997},
there exist $\bar{\xi}^{n}\in{\rm
conv}(\xi^{n},\xi^{n+1},\dots)$, $n\geq{}1$, and a nonnegative supermartingale
$\{\widetilde{\xi}_{t}\}_{t\geq{}0}$ with
$\widetilde\xi_{0}\leq{}1$ such that
$\{\bar{\xi}^{n}\}_{n\geq{}1}$ is \emph{Fatou
convergent to $\widetilde{\xi}$ on the rational
numbers $\pi$}; namely,
\begin{equation}\label{FatouConv}
    \widetilde\xi_{t}=\limsup_{s\downarrow{}t\;:\;
    s\in\pi}\limsup_{n\rightarrow\infty} \bar\xi^{n}_{s}=
    \liminf_{s\downarrow{}t\;:\;
    s\in\pi}\liminf_{n\rightarrow\infty}
    \bar\xi^{n}_{s},\quad a.s.
\end{equation}
for all $t\geq{}0$.
By Fatou's Lemma, it is not hard to check that
$\left\{\tilde\xi(t)
V_{t}\right\}_{t\leq{}T}$ is a
supermartingale for every admissible portfolio with value
process $V$, and hence,  $\widetilde{\xi}\in\widetilde{\Gamma}$.
Next, since the $\xi^{n}$'s are constant on $[T,\infty)$ and
$\widetilde{U}(\cdot;\omega)$ is convex, Fatou's Lemma implies
that
\begin{align*}
    \bbe\left[\widetilde{U}(y\tilde\xi_{_{T}}
    ,\omega)\right]
    &\leq v_{_{\Gamma}}(y),
\end{align*}
Finally, we need to verify  that, when $\Gamma$ is convex,
equality above is attained and
and that $\widetilde\xi$ can be approximated by elements of
$\Gamma$. Both facts are clear since
$\{\bar\xi^{n}\}\subset\Gamma$ and
$\lim_{n\rightarrow\infty}\bar\xi^{n}_{_{T}}=
\widetilde{\xi}_{_{T}}$ a.s. Then, by the
continuity and boundedness of $\widetilde{U}$,
    \[
        v_{_{\Gamma}}(y)\leq{}
        \lim_{n\rightarrow\infty}
        \bbe\left[\widetilde{U}(y\bar\xi^{n}_{_{T}}B_{T}^{-1})\right]
        =\bbe\left[\widetilde{U}(y\widetilde\xi_{_{T}}B_{T}^{-1},\omega)\right].
    \]
\end{proof}
}

We now give a result that is crucial for proving the strong duality in (\ref{WeakDualityEq}).
\begin{thrm}\label{MainThrm1}
    Suppose that (\ref{CondOnClassOfSprMart1}) is satisfied
    and $\Gamma$ is convex.
    Then, for any $z\in\left(0,w_{_{\Gamma}}\right)$,
    there exist $y(z)>0$ and
    $\xi^{*}_{y(z)}\in\widetilde\Gamma$ such that
\begin{enumerate}
    \item[(i)] $\bbe\left[\widetilde{U}\left(
    y(z) \xi^{*}_{y(z)}(T)
    ,\omega\right)\right]
        \leq \bbe\left[\widetilde{U}\left(
        y(z)\xi(T)
        ,\omega\right)\right],
        \quad { \forall \quad \xi\in\Gamma}$;
    \item[(ii)]
    $\bbe\left[V^{\Gamma}_{z}\,
        \xi^{*}_{y(z)}(T)
        \right]=z$,
    where
    \[
        V^{\Gamma}_{z}:= I\left(
        y(z){ \xi^{*}_{y(z)}}(T)
        \right)\wedge H;
    \]
    \item[(iii)]
    $u(z)\leq{} \bbe\left[U\left(V^{\Gamma}_{z};\omega\right)\right]$.
\end{enumerate}
\end{thrm}

\begin{proof}
    We borrow the idea of  \cite{Karatzas:1991}. For simplicity let us write $v(y)$
    instead of $v_{_{\Gamma}}(y)$.
    Recall that
    $w_{_{\Gamma}}:=\sup_{\xi\in\Gamma}\bbe\left[\xi(T)
    H\right]$
    and define $v(0):=\bbe\left[U(H;\omega)\right]$.
    In light of Lemma \ref{ProptyDual}, the continuous function
    $f_{z}(y):=v(y)+zy$ satisfies
    \[
        \lim_{y\downarrow{}0}\frac{f_{z}(y)-f_{z}(0)}{y}=-w_{_{\Gamma}}+z
        <0,\quad{\rm and}\quad
        f_{z}(\infty)=\infty,
    \]
    for all $z < w_{_{\Gamma}}$. Thus, $f_{z}(\cdot)$ attains its
    minimum at some $y(z)\in(0,\infty)$.
    By Proposition \ref{ProptyDual},
    we can find a $\xi_{y(z)}\in\widetilde\Gamma$ such
    that
    \[
        v(y(z))=\bbe\left[\widetilde{U}(y(z)\xi_{y(z)}(T)
        ,\omega)\right],
    \]
    proving the (i) above.
    Now, consider the function
    \[
        F(u):=u y(z)z+
        \bbe\left[\widetilde{U}\left(uy(z)\xi_{y(z)}(T)
        \right)\right],
        \quad u>0.
    \]
    Since $\xi_{y(z)}$ can be approximated by
    elements in $\Gamma$, for each $\varepsilon>0$ there exists a
    $\xi^{y,\varepsilon}_{y(z)}\in\Gamma$ such that
    \[
        \bbe\left[\widetilde{U}\left(y\xi_{y(z)}(T)
        \right)\right]
        >\bbe\left[\widetilde{U}\left(y\xi^{y,\varepsilon}_{y(z)}(T)
        \right)\right]-\varepsilon.
    \]
    It follows that for each $\varepsilon>0$,
    \begin{align*}
        { \inf_{u>0} \, F(u)}
        &\geq\inf_{y>0}
        \left\{ yz+
        \bbe\left[\widetilde{U}\left(y\xi^{y,\varepsilon}_{y(z)}(T)
        \right)\right]\right\}
        -\varepsilon\\
        &\geq\inf_{y>0}
        \left\{ yz+v(y)\right\}-\varepsilon\\
        &=y(z)z+
        \bbe\left[\widetilde{U}\left(y(z)\xi_{y(z)}(T)
        \right)\right]
        -\varepsilon.
    \end{align*}
    Since $\varepsilon>0$ is arbitrary, the function $F(u)$
    attains its minimum at $u=1$. On the other hand,
    $\frac{F(1+h)-F(1)}{h}$ equals
    \begin{align*}
        y(z)z+\bbe\left[
        \frac{\widetilde{U}((1+h)y(z)\xi_{y(z)}(T)
        )-
        \widetilde{U}(y(z)\xi_{y(z)}(T)
        )}{h}\right],
    \end{align*}
    which converges to
    \[
        y(z)z-y(z)\bbe\left[\left(I\left(
        y(z)\xi_{y(z)}(T)
        \right)\wedge H\right)
        \xi_{y(z)}(T)
        \right]
    \]
    as $h\rightarrow{}0$. Here, we use (\ref{DerConvexDual})
    and the dominated convergence theorem.
    Then,
    \[
        \bbe\left[\left(I\left(
        y(z)\xi_{y(z)}(T)
        \right)\wedge H\right)
        \xi_{y(z)}(T)
        \right]=z.
    \]
    This proves (ii) of the theorem, and also, (iii) in light
    of (\ref{Utilde}) and (\ref{dualrel}).
\end{proof}
We note that Theorem \ref{MainThrm1} essentially provides an
upper bound for the optimal final utility of the form
$\bbe\left[U\left(V^{\Gamma}_{z};\omega\right)\right]$, for certain
``reduced" contingent claim $V^{\Gamma}_{z}\leq{}H$. By suitably
choosing the dual class $\Gamma$, we shall prove in the next two
sections that this reduced contingent claim is (super-) replicable
with an initial endowment $z$.

\section{Characterization of the optimal dual}
\label{SpecDualClssSect}
 \setcounter{equation}{0}

We now give a full description of a dual class $\Gamma$ for which
\emph{strong duality}, i.e., $u(z)=v_{_{\Gamma}}(y)+zy$, holds.
 Denote $\calV^{+}$ to be the class of all
real-valued c\`adl\`ag, non-decreasing, adapted processes $A$ null
at zero. We will call such a process ``increasing". In what follows we let $\calE(X)$ be the Dol\'eans-Dade
stochastic exponential of the semimartingale $X$ (see e.g.
\cite{Shiryaev:2003} for their properties).
Let
    \begin{equation}\label{BscClass3}
    {\calS}:=\left\{X_{t}:=
    \int_{0}^{t}G(s)d W_{s}
    +\int_{0}^{t}\int_{\bbr_{0}}F(s,z)\widetilde{N}(ds,dz):
    F\geq{}-1\right\},
    \end{equation}
and {consider} the associated class of exponential local supermartingales:
\begin{equation}\label{ExpSmmartClassb}
    \Gamma(\calS):= \{ \xi:=\xi_{0}\calE(X- A):
    X\in\calS, A\text{ increasing, and }\xi\geq{}0\}.
\end{equation}
In (\ref{BscClass3}), we assume that $G\in L_{loc}^{2}(W)$,
$F\in G_{loc}(N)$, and that $F(t,\cdot)=G(t)=0$, for all $t\geq{}T$.
The following result shows not only that the class
\begin{equation}\label{MainDualClass}
    \Gamma:=\widetilde\Gamma\cap \Gamma(\calS),
\end{equation}
is convex, but {also that} the dual optimum, whose existence is {deduced} from
Theorem \ref{MainThrm1}, remains in $\Gamma$. The proof of this
result is based on a powerful representation for nonnegative
supermartingales due to F\"ollmer and Kramkov \cite{Follmer:1997}
(see Theorem \ref{ThrmFollmerKramkov} in the appendix), and a technical result {about} the closedness of the class of
integrals with respect to Poisson random measures, under \'Emery's
topology. We shall differ {the presentation of these two fundamental results} to
Appendix \ref{CovClssSect} in order to {continue} with 
our discussion {of the dual problem}.

\begin{thrm}\label{ExistenceDual}
The class $\Gamma$ is convex, and
if (\ref{CondOnClassOfSprMart1}) is satisfied,
the dual optimum $\xi^{*}_{y(z)}$ of
Theorem \ref{MainThrm1} belongs to $\Gamma$,
for any $0<z<w_{_{\Gamma}}$.
\end{thrm}

\begin{proof}
Let us check that $\calS$ meets with
the conditions in Theorem \ref{ThrmFollmerKramkov}.
Indeed, each
$X$ in $\calS$ is locally bounded from below
since, defining
\(
    \tau_{n}:=\inf\{t\geq{}0: X_{t}<-n\},
\)
\[
    X^{\tau_{n}}_{t}\geq X_{\tau_{n}^{-}} -(\Delta X_{\tau_{n}})^{-}
    {\bf 1}_{\tau_{n}<\infty}\geq -n-1,
\]
where $(x)^{-}=-x{\bf 1}_{x<0}$.
Condition (i) of Theorem \ref{ThrmFollmerKramkov}
is straightforward, while condition
(ii) follows from Theorem \ref{ClosenessResult}. Finally,
condition (iii) holds because the processes in
$\calS$ are already local martingales with respect to
$\bbp$ and hence $\bbp\in\calP(\calS)$ with
$A^{\calS}(\bbp)\equiv 0$.
By the Corollary \ref{ClosednessClasses},
we conclude that $\Gamma(\calS)$
is convex and closed under Fatou convergence on dense countable
sets.
On the other hand, $\widetilde\Gamma$
is  also convex and closed under Fatou convergence,
and thus so is the class
$\Gamma:=\widetilde\Gamma\cap \Gamma(\calS)$.
To check the second statement,
recall that the existence of the dual minimizer
$\xi^{*}_{y(z)}$ in Theorem \ref{MainThrm1} is
guaranteed from Proposition \ref{ProptyDual},
where it is seen that $\xi^{*}_{y(z)}$ is the
Fatou limit of a sequence in $\Gamma$ (see the
proof of Proposition \ref{ProptyDual}).
This suffices to conclude that
$\xi^{*}_{y(z)}\in\Gamma$ since $\Gamma$ is closed under
under Fatou convergence.
\end{proof}

In the rest of this section, we present some properties of the
elements in $\Gamma$ and of the dual optimum $\xi^{*}\in\Gamma$. In
particular, conditions on the ``parameters'' $(G,F,A)$ so that
$\xi\in\Gamma(\calS)$ is in $\widetilde\Gamma$ are established.
First, we note that without lose of generality, $A$ can be assumed
predictable.
\begin{lmma}\label{lmma1SectDualClass}
Let
\begin{equation}\label{ExpSemimrtgl}
    \xi:=\xi_{0}\calE(X- A)\in \Gamma(\calS).
\end{equation}
Then, there exist a predictable process $A^{p}\in\cal{V}^{+}$
and {a process} $\widehat{X}\in\calS$ such
    that $\xi=\xi_{0}\calE(\widehat{X}- A^{p})$.
\end{lmma}
\begin{proof}
Let $X_{t}:= \int_{0}^{t}G(s)d W_{s}
+\int_{0}^{t}\int_{\bbr_{0}}F(s,z)\widetilde{N}(ds,dz)\in\calS$.
Since $F\in G_{loc}(N)$, there are stopping times
$\tau'_{n}\nearrow\infty$ such that
\[
    \bbe\int_{0}^{\tau'_{n}}\int_{\bbr}|F(s,z)|{\bf 1}_{|F|>1}\nu(dz)ds<\infty;
\]
cf. Theorem II.1.33 in \cite{Shiryaev:2003}.
Now, define
\[
    \tau''_{n}:=\inf\{t\geq{}0: A_{t}>n\},
\]
and $\tau_{n}:=\tau'_{n}\wedge\tau''_{n}$. Then,
\begin{align*}
    \bbe A^{\tau_{n}}_{\infty}
    &=\bbe [A_{\tau_{n}^{-}}]+
    \bbe \left[\Delta A_{\tau_{n}}\right]\leq n+1+
    \bbe\left[|F(\tau_{n},Z_{\tau_{n}})|\right]\\
    &\leq n+2+\bbe\int_{0}^{\tau_{n}}\int_{\bbr}|F(s,z)|{\bf 1}_{|F|>1}\nu(dz)ds<\infty,
\end{align*}
where we used that $\Delta X_{t}-\Delta A_{t}\geq{}-1$.
Therefore, $A$ is locally integrable, increasing, and thus, its
predictable compensator $A^{p}$ exists. Now, by the representation
theorem for local martingales (see Theorem III.4.34
\cite{Shiryaev:2003}), the local martingale $X':=A-A^{p}$
admit the representation
\[
    X'_{\cdot}:=
    \int_{0}^{\cdot}G'(s)d W_{s}
    +\int_{0}^{\cdot}\int_{\bbr_{0}}F'(s,z)\widetilde{N}(ds,dz).
\]
Finally,
\(
    \xi=\xi_{0}\calE(X- A)=
    \xi_{0}\calE(X-X'-A^{p}).
\)
The conclusion of the proposition
follows since $\widehat{X}:=X-X'$ is necessarily
in $\calS$.
\end{proof}

The following result gives necessary conditions
for a process $\xi\in\Gamma(\calS)$ to belong to
$\widetilde\Gamma$.
Recall that a predictable increasing process
$A$ can be uniquely decomposed as the sum of three
predictable increasing processes,
 \begin{equation}\label{DecmpIncr}
    A=A^{c}+A^{s}+A^{d},
 \end{equation}
 where $A^{c}$ is the {\it absolutely continuous} part,
 $A^{s}$ is the {\it singular continuous} part,
 and  $A^{d}_{t}=\sum_{s\leq{}t}\Delta A_{s}$
 is the jump part (cf. Theorem 19.61 in \cite{Hewitt}).
\begin{prop}\label{BudgetConstraintResult}
    Let  $\xi:=\xi_{0}\calE(X- A)\geq{}0$,
    where $\xi_{0}>0$,
    \[
        X_{t}:=
        \int_{0}^{t}G(s)d W_{s}
        +\int_{0}^{t}\int_{\bbr_{0}}F(s,z)\widetilde{N}(ds,dz)
        \in\calS,
    \]
    and $A$ is an increasing predictable process.
    Let $\tau$ be the ``sinking time''
    of the supermartingale $\xi$:
    \[
        \tau:=\sup_{n}\inf\{t:\xi_{t}<\frac{1}{n}\}
        =\inf\{t:\Delta X_{t}=-1\quad {\rm or}\quad
        \Delta A_{t}=1\}.
    \]
    Also, let $a_{t}=\frac{d A^{c}_{t}}{d  t}$.
    Then,
    \(
        \displaystyle\left\{\xi_{t}S_{t}\right\}_{t\leq{}T}
    \)
    is a supermartingale if and only if
    the following two conditions are satisfied:
    \begin{enumerate}
    \item[(i)]  There exist stopping times
    $\tau_{n}\nearrow\tau$ such that
    \begin{equation}
        \label{IntCondForHAdpt5}
        \bbe\int_{0}^{\tau_{n}}\int_{\bbr}|v(s,z) F(s,z)|\nu(dz)ds<\infty.
    \end{equation}
    \item[(ii)] For $\bbp$-a.e. $\omega\in\Omega$,
    \begin{equation}
        \label{IntCondForHAdpt4}
        h_{t}\leq a_{t},
    \end{equation}
    for almost every $t\in[0,\tau(\omega)]$, where
    \[
        h_{t}:=b_{t}+\sigma_{t} G(t)
        +\int_{\bbr}v(t,z) F(t,z)\nu(dz).
    \]
    \end{enumerate}
\end{prop}

\begin{proof}
    Recall that $\xi$ and $S$ satisfies the SDE's
    \begin{align*}
        d\xi_{t}&= \xi_{t^{-}} \left(G(t)\, dW_{t}+
        \int_{\bbr_{0}}F(t,z) \widetilde{N}(dt,dz)
        -d A_{t}\right),\\
        d {S}_{t}&=
        {S}_{t^{-}}\left(b_{t}\,dt+
        \sigma_{t}\, d W_{t}+
        \int_{\bbr_{0}} v(t,z)\, \widetilde{N}(dt,dz)
        \right).
    \end{align*}
    Integration by parts and the predictability of $A$
    yield that
        \begin{align}
        \xi_{t}{S}_{t}&=\text{local martingale}+
        \int_{0}^{t} b_{s}\,\xi_{s^{-}}{S}_{s^{-}}ds
        +\int_{0}^{t}\sigma_{s} G(s)\,\xi_{s^{-}}{S}_{s^{-}}ds\nonumber\\
        &~~
        -\int_{0}^{t} \xi_{s^{-}}{S}_{s^{-}}\,d A_{s}
        +\int_{0}^{t}\int_{\bbr_{0}} v(s,z) F(s,z) \xi_{s^{-}}{S}_{s^{-}}N(ds,dz).
        \label{IntgrByPart}
    \end{align}
    Suppose that $\{\xi_{t}{S}_{t}\}_{t\geq{}0}$ is a nonnegative supermartingale.
    Then, the integral $\int_{0}^{t}\int_{\bbr}
    v(s,z) F(s,z) \xi_{s^{-}}{S}_{s^{-}}N(ds,dz)$
    must have locally integrable variation in light of the
    Doob-Meyer decomposition for supermartingale
    (see e.g. Theorem III.13 in \cite{Protter}).
    Therefore, there exist stopping times
    $\tau^{1}_{n}\nearrow\infty$ such that
    \[
        \bbe\int_{0}^{\tau^{1}_{n}}\int_{\bbr} |v(s,z) F(s,z) \xi_{s^{-}}{S}_{s^{-}}|\nu(dz)ds<\infty.
    \]
    Then, (i) is satisfied with
    $\tau_{n}:=\tau^{1}_{n}\wedge \tau^{2}_{n}\wedge\tau^{3}_{n}$,
    where $\tau^{2}_{n}:=\inf\{t:\xi_{t}<\frac{1}{n}\}$ and
    $\tau^{3}_{n}:=\inf\{t:\widetilde{S}_{t}<\frac{1}{n}\}$.
    Next, we can write (\ref{IntgrByPart}) as
    \[
        \xi_{t}{S}_{t}=\text{local martingale}
        -\int_{0}^{t} \xi_{s^{-}}{S}_{s^{-}}(d A_{s}-h_{s}ds).
    \]
    By the Doob-Meyer representation for supermartingales and the uniqueness of
    the canonical decomposition for special semimartingales,
    the last integral must be increasing. Then,
    $a_{t}\geq{}h_{t}$ for $t\leq{}\tau$
    since $\xi_{t^{-}}>0$ and $\xi_{t}=0$ for $t\geq{}\tau$
    (see I.4.61 in \cite{Shiryaev:2003}).

    We now turn to the sufficiency of conditions (i)-(ii).
    Since
    $\{\xi_{t^{-}}{S}_{t^{-}}\}_{t\geq{}0}$ is locally
    bounded,
    \[
        \int_{0}^{t}\int_{\bbr} |v(s,z) F(s,z)| \xi_{s^{-}}
        {S}_{s^{-}}
        {\bf 1}_{s\leq\tau_{n}}\nu(dz)ds
    \]
    is locally integrable. Then, from (\ref{IntgrByPart}),
    we can write
    \[
        \xi_{t\wedge\tau_{n}}{S}_{t\wedge\tau_{n}}
        = \text{local martingale}
        -\int_{0}^{t} \xi_{s^{-}}{S}_{s^{-}}
        {\bf 1}_{s\leq\tau_{n}}(d A_{s}-h_{s}ds).
    \]
    Condition (ii) implies that $\{\xi_{t\wedge\tau_{n}}{S}_{t\wedge\tau_{n}}\}$
    is a supermartingale, and by Fatou,
    $\{\xi_{t\wedge\tau}{S}_{t\wedge\tau}\}_{t\geq{}0}$
    will be a supermartingale.
    This concludes the prove since $\xi_{t}=0$ for $t\geq{}\tau$,
    and thus, $\xi_{t\wedge\tau}{S}_{t\wedge\tau}=\xi_{t}{S}_{t}$,
    for all $t\geq{}0$.
    \end{proof}

The following result gives sufficient
and necessary conditions for $\xi\in\Gamma(\calS)$
to belong to $\widetilde{\Gamma}$.
Its proof is similar to that of Proposition \ref{BudgetConstraintResult}.
\begin{prop}\label{ConditionsLocalSuper}
    Under the setting and notation of
    Proposition \ref{BudgetConstraintResult},
    $\xi\in\Gamma(\calS)$ belongs
    to $\widetilde{\Gamma}$ if and only if
    condition (i) in Proposition \ref{BudgetConstraintResult}
    holds and, for any locally bounded admissible
    trading strategies $\beta$,
    \begin{equation}
        \label{IntCondForHAdpt4b}
        \bbp\left[\left\{\omega:
        h_{t}\beta_{t}\leq
        a_{t},\text{ for
        a.e. }t\in[0,\tau(\omega)]\right\}\right]=1.
    \end{equation}
\end{prop}

The previous result can actually be made more explicit
under additional information on the structure of the
jumps exhibited by the stock price process.
We consider two cases: {when the jumps come from the
superposition of shot-noise Poisson processes, 
and when the random field $v$ exhibit
a multiplicative structure}.
Let us first extend Proposition \ref{CndAdmissibility} {in these two cases}.
\begin{prop}\label{TwoCases}
    (i) Suppose $\nu$ is atomic with finitely many atoms
    $\{z_{i}\}_{i=1}^{k}$. Then,
    a predictable locally bounded strategy $\beta$ is admissible
    if and only if $\bbp\times dt$-a.e.
    \[
        -\frac{1}{\max_{i} v(t,z_{i})\vee 0}\leq\beta_{t}\leq
        -\frac{1}{\min_{i} v(t,z_{i})\wedge 0}.
    \]
    (ii) Suppose that
        \(
        v(t,z)=\zeta_{t}\vartheta(z),
        \)
    for a predictable {locally bounded} process $\zeta$ {such that 
    $\bbp\times dt$-a.e. $\zeta_{t}(\omega)\neq{}0$ and $\zeta_{t}^{-1}$ is locally bounded}, and a deterministic function $\vartheta$ {such that $\nu(\{z:\vartheta(z)=0\})=0$}. Then, a predictable locally bounded strategy $\beta$ is admissible
    if and only if $\bbp\times dt$-a.e.
    \[
        -\frac{1}{{\bar{\vartheta}}\vee 0}\leq\beta_{t}\,\zeta_{t}\leq
        -\frac{1}{{\underline{\vartheta}}\wedge 0},
    \]
    where {$\bar{\vartheta}:=\sup\{ \vartheta(z): z\in{\rm supp}(\nu)\}$} and 
     {$\underline{\vartheta}:=\inf\{ \vartheta(z): z\in{\rm supp}(\nu)\}$}.
\end{prop}
\begin{proof}
   {From Proposition \ref{CndAdmissibility}}, recall that $\bbp$-a.s.
    \[
        \beta_{t} v(t,\Delta  Z_{t})\geq{}-1,
    \]
    for a.e. $t\leq{}T$. Then, for any closed set $C\subset\bbr_{0}$,
    $0\leq s< t$, and $A\in\calF_{s}$,
    \[
        \sum_{s<u\leq{}t}\chi_{A}(\omega)\,\chi_{C}(\Delta Z_{u})\left\{\beta_{u}v(u,\Delta Z_{u})
        +1\right\}\geq{}0.
    \]
    Taking expectation, we get
    \[
        \bbe\int_{s}^{t}\chi_{A}\int_{C}
        \left\{\beta_{u}v(u,z)
        +1\right\}\nu(dz)du\geq{}0.
    \]
    Since such processes $H_{u}(\omega):=
    \chi_{A\times(s,t]}(\omega,u)$ generate the class of predictable
    processes, we conclude that $\bbp\times dt$-a.e.
    \[
        -1\leq\beta_{t}\;\frac{\int_{C} v(t,z)\nu(dz)}{\nu(C)}.
    \]
    Let us prove (ii) (the proof of (i) is similar).
    Notice that
    \[
        \inf_{z\in U} \vartheta(z)=\inf_{C\subset\bbr_{0}} \frac{\int_{C} \vartheta(z)\nu(dz)}{\nu(C)}\leq
        \sup_{C\subset\bbr_{0}} \frac{\int_{C} \vartheta(z)\nu(dz)}{\nu(C)}
        = \sup_{z\in{}U} \vartheta(z),
    \]
    where $U$ is the support of $\nu$.
    Suppose that $\inf_{z\in U} \vartheta(z)<0<\sup_{z\in U} \vartheta(z)$.
    Then, by considering closed sets $C_{n},C_{n}^{'}\subset \bbr_{0}$ such that
    \[
        \frac{\int_{C_{n}} \vartheta(z)\nu(dz)}{\nu(C_{n})}
        \nearrow \sup_{z} \vartheta(z),\quad
        \text{ and }\quad
        \frac{\int_{C_{n}^{'}} \vartheta(z)\nu(dz)}{\nu(C_{n}^{'})}
        \searrow \inf_{z} \vartheta(z),
    \]
    as $n\rightarrow\infty$, we can prove the necessity.
    The other two cases (namely,
    $\inf_{z} \vartheta(z)\geq 0$ or
    $0\geq\sup_{z} \vartheta(z)$) are
    proved in a similar way. Sufficiency follows since,
    $\bbp$-a.s.,
    \begin{align*}
        \{t\leq{}T:\beta_{t}\zeta_{t}v\left(\Delta Z_{t}\right)<-1\}
        &\subset
        \{t\leq{}T:\beta_{t}\zeta_{t}\sup_{z\in U} v(z)<-1\}\; \cup\\
        &\quad\{t\leq{}T:\beta_{t}\zeta_{t}\inf_{z\in{}U} v(z)<-1\}.
    \end{align*}
\end{proof}

\begin{exmpl}\label{Exmpl}
 {\rm   It is worth pointing out some consequences:
    \begin{enumerate}
    \item[(a)]
    In the time homogeneous case, where $v(t,z)=z$,
    the extreme points of the support of $\nu$
    (or what accounts to the same, the infimum and supremum of
    all possible jump sizes) determine completely
    the admissible strategies. For instance, if the
    L\'evy process can exhibit arbitrarily large
    or arbitrarily close to $-1$ jump sizes, then
    \[
        0\leq{}\beta_{t}\leq{}1;
    \]
    a constraint that can be interpreted
    as absence of shortselling and bank borrowing
    (this fact was already pointed out by Hurd \cite{Hurd:2004}).

    \item[(b)] In the case that
    ${\underline{\vartheta}\geq{}0}$, the admissibility condition
    takes the form
    \(
        -1/{\bar{\vartheta}}\leq\beta_{t}\,\zeta_{t}.
    \)
    If in addition $\zeta_{\cdot}<{}0$ (such that the stock prices
    exhibits only downward sudden movements), then
    \(
        -{1/(\bar{\vartheta}\zeta_{t})}\geq\beta_{t},
    \)
    and $\beta_{\cdot}\equiv -c$, with $c>0$ arbitrary, is admissible.
    In particular, {from Proposition \ref{ConditionsLocalSuper}}, if $\xi\in\Gamma(\calS)$ belongs to
    $\widetilde{\Gamma}$, then a.s.
    \(
        h_{t}\beta_{t}\leq{}a_{t}, \;\text{ for a.e. t }
        \leq\tau.
    \)
    {This means that $\xi\in \Gamma(\calS)\cap \widetilde\Gamma$ if and only if 
    condition (i) in Proposition \ref{BudgetConstraintResult}
    holds and 
    $\bbp-$a.s. $h_{t}\geq{}0$, for a.e.
     $t\leq{}\tau$}. 
     For a general $\zeta$ and still assuming that ${\underline{\vartheta}\geq{}0}$, {it follows that} $\beta$ admissible
    and $\xi\in\Gamma(\calS)\cap\widetilde{\Gamma}$
     satisfy that $\bbp$-a.s. 
    \[
         -\frac{1}{{\bar{\vartheta}}(\zeta_{t}\vee 0)}\leq
         \beta_{t}\leq
         -\frac{1}{{\bar{\vartheta}}(\zeta_{t}\wedge 0)},
         \quad
         {\rm and}\quad
         h_{t}\zeta_{t}^{-1}{\bf 1}_{\{t\leq\tau\}}\leq 0,
    \]
    for a.e. $t\geq{}0$.\qed
    \end{enumerate}
    }
\end{exmpl}

We now extend Proposition \ref{ConditionsLocalSuper}
in the two cases introduced in Proposition \ref{TwoCases}.
{Its proof follows from Propositions \ref{ConditionsLocalSuper} and \ref{TwoCases}.}
\begin{prop}\label{TwoCasesb}
{Suppose that either (i) or (ii) in Proposition \ref{TwoCases} are satisfied, in which case, define:
\[
        \widehat{h}_{t}:=
        \left\{\begin{array}{ll}
        -\frac{h_{t}}{\max_{i} v(t,z_{i})\vee 0}
        {\bf 1}_{\{h_{t}<0\}}
        -\frac{h_{t}}{\min_{i} v(t,z_{i})\wedge 0}{\bf 1}_{\{h_{t}>0\}},
        &\text{if (i) holds true},\\
        \\
        -\frac{h_{t}\zeta_{t}^{-1}}{\bar{ \vartheta}\vee 0}
        {\bf 1}_{\{h_{t}\zeta_{t}^{-1}<0\}}
        -\frac{h_{t}\zeta_{t}^{-1}}{\underline{\vartheta}\wedge 0}
        {\bf 1}_{\{h_{t}\zeta_{t}^{-1}>0\}}, &  \text{if (ii) holds true}.
        \end{array}\right.
    \]
 Then, 
a process $\xi\in\Gamma(\calS)$ belongs to $\widetilde{\Gamma}$ if and only if
{condition (i) in Proposition \ref{BudgetConstraintResult} holds, and}
for $\;\bbp$-a.e. $\omega$,  
$\widehat{h}_{t}(\omega){\bf 1}_{\{t\leq{}\tau(\omega)\}}\leq a_{t}(\omega){\bf 1}_{\{t\leq{}\tau(\omega)\}}$, for a.e. $t\geq{}0$.}
\end{prop}
{
We remark that the  cases $\underline{\vartheta}\geq 0$ and $\bar{\vartheta}\leq 0$
do not lead to any absurd in the definition of $\hat{h}$ above as we are using the convention that $0\cdot \infty=0$. Indeed, for instance, if  $\underline{\vartheta}\geq 0$, it was seeing that $h_{t}\zeta_{t}^{-1}\leq 0$, for a.e. $t\leq \tau$, and thus, we set the second term in the definition of $\hat{h}$ to be zero.
}

Now we {can give a more explicit characterization} of
the dual solution $\xi^{*} =\calE(X^{*}-A^{*})$
to the problem (\ref{DefDualProblem}), which existence
was established in Proposition \ref{ExistenceDual}.
For instance, we will see that  $A^{*}$ is
absolutely continuous up to a predictable stopping time.
Below, we refer to Proposition
\ref{BudgetConstraintResult} for the notation.

\begin{prop}\label{AbsContDual}
Let $\xi:=\xi_{0}\,\calE(X-A)\in\Gamma(\calS)$,
\(
    \tau_{_{A}}:=\inf\{t:\Delta A_{t}=1\},
\)
and $\widetilde{A}_{t}:=\int_{0}^{t}a_{s}\,ds
    +{\bf 1}_{\{t\geq{}\tau_{_{A}}\}}$.
The followings two statements hold true:
\begin{enumerate}
    \item[(1)]
    $\widetilde{\xi}:=\xi_{0}\calE\left(X-\widetilde{A}\right)\geq{}\xi$.
    Furthermore, $\xi\in\widetilde{\Gamma}$ if and only if
    $\widetilde\xi\in\widetilde{\Gamma}$.
    \item[(2)]
    Suppose that either of the two conditions in Proposition
    \ref{TwoCasesb} are satisfied and denote
    \[
        \widehat{A}_{t}:=
        \int_{0}^{t}\widehat{h}_{s}{\bf 1}_{s\leq{}\tau}\,ds
        +{\bf 1}_{\{t\geq{}\tau_{_{A}}\}},
    \]
    where $\widehat{h}$ is defined accordingly to the assumed case.
    Then, $\xi_{\cdot}\leq{}\widehat{\xi}_{\cdot}$, and
    furthermore, the process
    $\widehat{\xi}:=\xi_{0}\calE(X-\widehat{A})$
    belongs to $\widetilde{\Gamma}$ if $\xi\in\widetilde{\Gamma}$.
\end{enumerate}
\end{prop}
\begin{proof}
    Let $A^{c},A^{s},A^{d}$ denote the increasing
    predictable processes in the decomposition (\ref{DecmpIncr}) of
    $A$. Since $A$ is predictable, there is no common
    jump times between $X$ and $A$. Then,
    \begin{align*}
        \xi_{t}&=\xi_{0}e^{X_{t}-A_{t}-\frac{1}{2}<X^{c},X^{c}>_{t}}
        \prod_{s\leq{}t}(1+\Delta X_{s})e^{-\Delta X_{s}}
        \prod_{s\leq{}t}(1-\Delta A_{s})e^{\Delta A_{s}}\\
        &\leq \xi_{0}e^{X_{t}-A^{c}_{t}-\frac{1}{2}<X^{c},X^{c}>_{t}}
        \prod_{s\leq{}t}(1+\Delta X_{s})e^{-\Delta X_{s}}{\bf 1}_{\{t<\tau_{_{A}}\}}
        =\widetilde{\xi}_{t},
    \end{align*}
    where we used that $A_{t}-\sum_{s\leq{}t}\Delta A_{s}=A^{c}_{t}+A^{s}_{t}
    \geq{}A^{c}_{t}$, and
    $\prod_{s\leq{}t}(1-\Delta A_{s})
    \leq {\bf 1}_{\{t<\tau_{_{A}}\}}$.
    Since both
    processes $\xi$ and $\widetilde{\xi}$ enjoy the same
    absolutely continuous part, and the same
    sinking time, the second statement in (1)
    is straightforward from Proposition \ref{ConditionsLocalSuper}.
    Part (2) follows from Proposition \ref{TwoCasesb} since
    the process
    $\widehat{a}_{t}:=\widehat{h}_{t}{\bf 1}_{t\leq{}\tau}$
    is nonnegative, predictable (since $h$ is predictable),
    and locally integrable (since
    $0\leq\widehat{h}\leq{}a$).
\end{proof}

{We remark that part (2) in Proposition \ref{AbsContDual} remains true if we take 
$\widehat{A}_{t}:=\int_{0}^{t}\widehat{h}_{s}{\bf 1}_{s\leq{}\tau_{_{A}}}\,ds
        +{\bf 1}_{\{t\geq{}\tau_{_{A}}\}}$.
        }
The following result is similar to Proposition 3.4 in
Xu \cite{Xu:2004} and {implies, in particular, that the optimum dual $\xi^{*}$ can be taken to be a local martingale}.
\begin{prop}\label{AbsContDual2}
{Suppose that either condition (i) or (ii) of Proposition \ref{TwoCases}
is satisfied. Moreover, in the case of condition (ii), assume additionally that 
\begin{equation}\label{CndLcDO}
	\nu\left(\{z\in{\rm supp}(\nu)\backslash\{0\}:\vartheta(z)=c\}\right)>0,
\end{equation}
for $c=\underline{\vartheta}$ if $\bar{\vartheta}>{}0$, and for $c=\bar{\vartheta}$ 
if $\underline{\vartheta}<0$.}
Let $\xi\in\widetilde{\Gamma}\cap \Gamma(\calS)$.
Then, there exists $\widetilde{X}\in\calS$ such that
$\widetilde{\xi}:=\xi_{0}\,\calE(\widetilde{X})\in\widetilde{\Gamma}$
and $\xi_{\cdot}\leq{}\widetilde{\xi}_{\cdot}$.
Furthermore, $\{\widetilde{\xi}(t)V^{\beta}_{t}\}_{t\leq{}T}$ is a local martingale
for all locally bounded admissible trading
strategies $\beta$.
\end{prop}
\begin{proof}
    Let us prove the case when condition (i) in Proposition \ref{TwoCases}
    is in force.
    In light of Proposition \ref{AbsContDual},
    we assume without loss of generality that
    \(
        A_{t}=\int_{0}^{t} a_{t}dt+{\bf 1}_{\{t\geq{}\tau_{_{A}}\}},
    \)
    with
    \(
        {a}_{t}:=\hat{h}_{t} {\bf 1}_{\{t\leq{}\tau\}}.
    \)
    Assume that $\min_{i} v(t,z_{i})<0<\max_{i} v(t,z_{i})$.
    Otherwise if, for instance, $\max_{i} v(t,z_{i})\leq{}0$, then it can be shown that
    $h_{t}\geq{}0$, a.s. {(similarly to case (b) in Example \ref{Exmpl})}, and the first term of $\hat{h}$ is $0$ under our
    convention that $\infty\cdot 0= 0$.
    Notice that, in any case,  one can find a predictable process
    $z$ taking values on $\{z_{i}\}_{i=1}^{n}$, such that
    \[
        \hat{h}_{t}=
        -\frac{h_{t}}{v(t,z(t))}.
    \]
    Write $\widetilde{X}_{\cdot}:=
    \int_{0}^{\cdot} G(s) d W_{s}+\int_{0}^{\cdot}\int_{\bbr_{0}} \widetilde{F}(s,z) d\widetilde{N}(s,z)$
    for a $\widetilde{F}\in G_{loc}(N)$ to be determined in the sequel.
    For $\widetilde{\xi}\geq{}\xi$
    it suffices to prove the existence of a field
    $D$ satisfying both conditions below:
    \[
        (a)\;D\geq{}0\quad{\rm and}\quad
        (b)\;\int_{\bbr_{0}} D(t,z)\nu(dz){\bf 1}_{\{t\leq{}\tau\}}\leq{}\hat{h}_{t}{\bf 1}_{\{t\leq{}\tau\}},
    \]
    (then, $\widetilde{F}$ is defined as $D+F$).
    Similarly, for $\widehat{\xi}$ to belong to
    $\widetilde{\Gamma}$ it suffices that
    \[
        (c)\; h_{t}+\int_{\bbr_{0}} v(t,z)D(t,z)\nu(dz)=0.
    \]
    Taking
    \[
        D(t,z):=-\frac{h_{t}}{v(t,z(t))\nu(\{z(t)\})}
        {\bf 1}_{\{z=z(t)\}},
    \]
    clearly non-negative,
    (b) and (c) hold with equality. Moreover,
    the fact that inequalities (c) hold with equality implies that
     $\{\widehat{\xi}(t)V^{\beta}_{t}\}_{t\leq{}T}$ is a local martingale
    for all locally bounded admissible trading
    strategy $\beta$ (this can be proved using the same
    arguments as in the sufficiency
    part of Proposition \ref{BudgetConstraintResult}).
    {Now suppose that condition (ii) in Proposition \ref{TwoCases} holds. For simplicity, let us assume that $\underline{\vartheta}<0<\bar{\vartheta}$ (the other cases can be analyzed following arguments similar to Example \ref{Exmpl}). Notice that 
    (\ref{CndLcDO}) implies the existence of a Borel $\underline{C}$ (resp. $\bar{C}$) 
    such that $\vartheta(z)\equiv \underline{\vartheta}$ on $\underline{C}$ 
    (resp. $\vartheta(z)\equiv \bar{\vartheta}$ on $\bar{C}$ ) 
    and $0<\nu(\underline{C}),\nu(\bar{C})<\infty$. Taking
    \[
        D(t,z):=-\frac{h_{t}\zeta_{t}^{-1}}{\bar{\vartheta}\nu(\bar{C})}
        {\bf 1}_{\bar{C}}(z){\bf 1}_{\{h_{t}\zeta_{t}^{-1}<0\}}
        -\frac{h_{t}\zeta_{t}^{-1}}{\underline{\vartheta}\nu(\underline{C})}
        {\bf 1}_{\underline{C}}(z){\bf 1}_{\{h_{t}\zeta_{t}^{-1}>0\}},
    \]
    (b) and (c) above will hold with equality.   }
    \end{proof}

\section{Replicability of the upper bound}
\label{RplcbltySect}

We now show that the tentative optimum final wealth
$V_{z}^{\Gamma}$, suggested by the inequality (iii) in Proposition
\ref{ExistenceDual}, is (super-) replicable. We will combine the
dual optimality of $\xi^{*}$ with the \emph{super-hedging theorem},
which states that given a contingent claim $\widehat{H}$ satisfying
$\bar{w}:=\sup_{\bbq\in\calM} \bbe_{\bbq}\{\widehat{H}\}<\infty$,
one can find for any fixed $z\geq{}\bar{w}$ an admissible trading
strategy $\beta^{*}$ (depending on $z$) such that
$V_{_{T}}^{z,\beta^{*}}\geq{}\widehat{H}$ almost surely (see Kramkov
\cite{Kramkov:1996}, and
 also Delbaen and Schachermayer \cite{DlbSch:1994}).
 Recall that $\calM$ denotes the class of all equivalent
risk neutral probability measures.
    \begin{prop}\label{AttainResult}
        Under the setting and conditions of Proposition
        \ref{ExistenceDual}, for any $0<z<w_{_{\Gamma}}$,
        there is an admissible trading strategy $\beta^{*}$
        for $z$  such that
        \[
            V_{_{T}}^{z,\beta^{*}}\geq
            I\left(
            y(z)\xi^{*}_{y(z)}(T)\right)\wedge H,
        \]
        and thus, the optimum of $u(z)$ is reached at
        the strategy $\beta^{*}$. In particular,
        \[
            V_{_{T}}^{z,\beta^{*}}=
            I\left(
            y(z)\xi^{*}_{y(z)}(T)\right),
        \]
        when
        $I\left(y(z)\xi^{*}_{y(z)}(T)\right)< H$.
    \end{prop}
\begin{proof}
    For simplicity, we write $\xi^{*}_{t}:=\xi^{*}_{y(z)}(t)$,
    $y=y(z)$, and
    \[
        V^{*}= I\left(
            y(z)\xi^{*}_{y(z)}(T)\right)\wedge H.
    \]
    Fix an equivalent risk neutral probability measure $\bbq\in\calM$,
    and let
    {$\xi'_{t}=\frac{d\, \bbq|_{\calF_{t}}}{d\, \bbp|_{\calF_{t}}}$
    be its corresponding density processes.
    Here, $\bbq|_{\calF_{t}}$ (resp. $\bbp|_{\calF_{t}}$)
    is the restriction of the measure $\bbq$
    (resp. $\bbp$) to the filtration $\calF_{t}$.}
    Under $\bbq$, $S_{\cdot}$ is a local martingale, and
    then, for any locally bounded $\beta$, $V_{\cdot}^{\beta}$
    is a $\bbq$-local martingale.
    By III.3.8.c in \cite{Shiryaev:2003},
    $\xi'V^{\beta}$ is a $\bbp$-local martingale
    (necessarily nonnegative by admissibility), and
    thus, $\xi'$ is in $\widetilde{\Gamma}$.
    On the other hand,
    $\xi'$ belongs to $\Gamma(\calS)$
    due to the exponential representation for
    \emph{positive} local martingales in Kunita \cite{Kunita:2004}
    (alternatively, by invoking Theorems III.8.3,
    I.4.34c, and III.4.34 in \cite{Shiryaev:2003},
    $\xi'\in\Gamma(\calS)$ even if $Z$ were just
    an additive process $Z$).
    By the convexity of the dual class $\Gamma=\Gamma(\calS)\cap \widetilde{\Gamma}$,
    $\xi^{(\varepsilon)}:=\varepsilon \xi' +(1-\varepsilon)\xi^{*}$
    belongs to $\Gamma$, for any $0\leq{}\varepsilon\leq{}1$. Moreover, since $\widetilde{U}$
    is convex and
    $\widetilde{U}'(y)=-(I(y)\wedge H)$,
    \[
        \left|\frac{\widetilde{U}\left(y\xi^{(\varepsilon)}_{_{T}}
        \right)-\widetilde{U}\left(y\xi^{*}_{_{T}}
        \right)}{\varepsilon}\right|
        \leq y 
        H \left|\xi'_{_{T}}-\xi^{*}_{_{T}}\right|
        \leq
        y 
        H \left(\xi'_{_{T}}+\xi^{*}_{_{T}}\right).
    \]
    The random variable
    $y 
    H \left(\xi'_{_{T}}+\xi^{*}_{_{T}}\right)$
    is integrable since by assumption $w_{_{\Gamma}}
    <\infty$. We can then apply dominated convergence
    theorem to get
    \[
        \lim_{\varepsilon\downarrow 0}
        \frac{1}{\varepsilon}\left\{
        \bbe\left[\widetilde{U}\left(y\xi^{(\varepsilon)}_{_{T}}
        \right)\right]-\bbe\left[\widetilde{U}\left(y\xi^{*}_{_{T}}
        \right)\right]\right\}
        =-y\bbe\left[
        V^{*}
        \left(\xi'_{_{T}}-\xi^{*}_{_{T}}\right)
        \right],
    \]
    which is nonnegative by condition (i) in Proposition
    \ref{ExistenceDual}. Then, using condition (ii) in Proposition
    \ref{ExistenceDual},
    \[
        \bbe_{\bbq}\left[
        V^{*}\right]
        =\bbe\left[
        V^{*}\xi'_{_{T}}\right]
        \leq\bbe\left[
        V^{*}\xi^{*}_{_{T}}
        \right]=z.
    \]
    Since $\bbq\in\calM$ is arbitrary,
    \(
        \sup_{\bbq\in\calM}
        \bbe_{\bbq}\left[
        V^{*}\right]
        \leq z.
    \)
    By \emph{the super-heading theorem},
    there is an admissible trading strategy $\beta^{*}$
    for $z$  such that
    \[
        V_{_{T}}^{z,\beta^{*}}\geq
        I\left(
        y(z)\xi^{*}_{y(z)}(T)
        \right)\wedge H.
    \]
    The second statement of the theorem
    is straightforward since $U(z)$
    is strictly increasing on $z<H$.
\end{proof}

\section{Concluding remarks}\label{FnlRmrksSect}

We conclude the paper with the following remarks.

\medskip
 {\bf (i) The dual class $\Gamma$.}~ The dual domain of the dual
problem can be taken to be the more familiar class of equivalent
risk-neutral probability measures $\calM$. To be more precise,
define
    \[
        \bar\Gamma:=\left\{\xi_{t}:=
        \frac{d\, \bbq|_{\calF_{t}}}{d\, \bbp|_{\calF_{t}}}
        :\bbq\in\calM\right\}.
    \]
    Since $\bar\Gamma$ is obviously a convex subclass of $\widetilde{\Gamma}$,
    Theorem \ref{MainThrm1} implies that,
    as far as
    \begin{equation}\label{SfCondExit}
        0<\bar{w}:=\sup_{\xi\in\bar\Gamma}
        \bbe\left[\xi_{_{T}}
        H\right]
        <\infty,
    \end{equation}
    for each $z\in(0,\bar{w})$, there exist
    $y:=y(z)>0$ and  $\xi^{*}:=\xi^{*}_{y(z)}\in\widetilde\Gamma$
    (not necessarily belonging to $\bar\Gamma$) such that (i)-(iii) in
    Proposition \ref{ExistenceDual}
    hold with $\Gamma=\bar{\Gamma}$. Finally, one can slightly modify the proof of
    Proposition \ref{AttainResult}, to conclude the replicability
    of
    \[
        V^{\bar\Gamma}_{z}:= I\left(
        y \xi^{*}_{_{T}}
        \right)\wedge H.
    \]
    Indeed, in the notation of the proof of the
    Proposition \ref{AttainResult}, the only step
    which needs to be justified in more detail is that
    \begin{equation}
    \label{NeededRel}
        \bbe\left[\widetilde{U}\left(y\xi^{*}_{_{T}}
        \right)\right]\leq \bbe\left[\widetilde{U}
        \left(y\xi^{(\varepsilon)}_{_{T}}
        \right)\right],
    \end{equation}
    for all $0\leq\varepsilon\leq{}1$, where
    $\xi^{(\varepsilon)}=\varepsilon \xi'+(1-\varepsilon) \xi^{*}$
    (here, $\xi'$ is a fixed element in $\bar\Gamma$).
    The last inequality follows from the fact that,
    by Proposition \ref{AttainResult} (c), $\xi^{*}$
    can be approximated by elements $\{\xi^{(n)}\}_{n\geq{}1}$ in
    $\bar{\Gamma}$ in the sense that
    $\xi^{(n)}_{_{T}}\rightarrow\xi^{*}_{_{T}}$ a.s.
    Thus,
    $\xi^{(\varepsilon)}$ can be approximated by
    the elements
    $\xi^{(\varepsilon,n)}:=
    \varepsilon \xi'+(1-\varepsilon) \xi^{(n)}$ in $\bar\Gamma$,
    for which we know that
    \[
        \bbe\left[\widetilde{U}\left(y\xi^{*}_{_{T}}
        \right)\right]\leq \bbe\left[\widetilde{U}
        \left(y\xi^{(\varepsilon,n)}_{_{T}}
        \right)\right].
    \]
    Passing to the limit as $n\rightarrow\infty$,
    we obtain  (\ref{NeededRel}).

    In particular we conclude that condition (\ref{SfCondExit}) is
    sufficient for both the existence of the solution to
    the primal problem and its characterization
    in terms of the dual solution
    $\xi^{*}\in\widetilde{\Gamma}$ of the dual problem
    induced by $\Gamma=\bar{\Gamma}$.
    We now further know that $\xi^{*}$ belongs to the class
    $\widetilde{\Gamma}\cap\Gamma(\calS)$ defined in
    (\ref{MainDualClass}), and hence,
    enjoys an explicit parametrization of the form
    \[
        \xi^{*}:=
        \calE\left(\int_{0}^{\cdot}G^{*}(s)d W_{s}
        +\int_{0}^{\cdot}\int_{\bbr_{0}}F^{*}(s,z)\widetilde{N}(ds,dz)
        -\int_{0}^{\cdot} a^{*}_{s} ds\right),
    \]
    for some triple $(G^{*},F^{*},a^{*})$.

\medskip

{\bf (ii) Market driven by general additive models.}~
Our analysis can be extended to more general multidimensional
models driven by additive processes (that is, processes with
independent, possibly non-stationary increments; cf. Sato
\cite{Sato} and Kallenberg \cite{Kallenberg}). For instance, let
$(\Omega,\calF,\bbp)$ be a complete probability space on which is
defined a $d$-dimensional additive process $Z$ with L\'evy-It\^o
decomposition:
\[
    Z_{t}=\alpha t + \Sigma\, W_{t}+
        \int_{0}^{t}\int_{\{\|z\|> 1\}}z
        N(ds,dz)+
        \int_{0}^{t}\int_{\{\|z\|\leq{}1\}}z
        \widetilde{N}(ds,dz),
\]
where $W$ is a standard $d-$dimensional Brownian motion,
$N(dt,dz)$ is an independent Poisson random measure on
$\bbr_{+}\times \bbr^{d}$, and
$\widetilde{N}(dt,dz)=N(dt,dz)-\bbe N(dt,dz)$.
Consider a market model consisting of $n+1$ securities: one risk free bond with
price
\[
    d B_{t} := r_{t} B_{t} dt,\quad B_{0}=1, \quad t\geq{}0,
\]
and $n$ risky assets with prices determined
by the following stochastic differential equations with jumps:
\[
  d S^{i}_{t} = S^{i}_{t^{-}}
    \left\{
    b^{i}_{t}\, dt+\sum_{j=1}^{d}\sigma^{ij}_{t}d W^{j}_{t}
    +\int_{\bbr^{d}}v^{i}(t,z) \widetilde{N}(ds,dz)\right\},
    \quad i=1,\dots,n,
\]
where the processes $r$, $b$, $\sigma$, and
$v$ are predictable satisfying usual integrability conditions
(cf. Kunita \cite{Kunita:2003}). We assume that
$\calF:=\calF_{\infty^{-}}$, where
$\bbf:=\left\{\calF_{t}\right\}_{t\geq{}0}$ is the natural
filtration generated by $W$ and $N$; namely,
$\calF_{t}:=\sigma(W_{s},N([0,s]\times A): s\leq{}t,
A\in\calB(\bbr^{d}))$. The crucial property, particular to this
market model, that makes our analysis valid is the representation
theorem for local martingales relative to $Z$
(see Theorem III.4.34 in \cite{Shiryaev:2003}).
The definition of the dual class $\Gamma$ given in
Section \ref{SpecDualClssSect} will remain unchanged, and
only very minor details will change in the
proof of Theorem \ref{ClosenessResult}.
Some of the properties of the results in
Section \ref{SpecDualClssSect} regarding the properties of
$\Gamma$ will also change slightly.
We remark that, by taking a real (nonhomogeneous) Poisson process,
the model and results of Chapter 3 in Xu \cite{Xu:2004}
will be greatly extended. We do not pursue the details here due to the limitation of the
length of this paper.

\medskip
{\bf (iii) Optimal wealth-consumption problem.}~ Another classical
portfolio optimization in the literature is that of optimal
wealth-consumption strategies under a budget constraint. Namely,
we allow the agent to spend money outside the market, while
maintaining ``solvency'' throughout $[0,T]$. In that case the
agent aims to maximize the cost functional that contains a ``{\it
running cost}":
\[
    \bbe\left[U_{1}(V_{_{T}})+\int_{0}^{T}U_{2}(t,c_{t})dt\right],
\]
where $c$ is the instantaneous rate of consumption.
To be more precise, the cumulative consumption at time $t$ is
given by $C_{t}:=\int_{0}^{t}c_{u}du$ and the (discounted)
wealth at time $t$ is given by
\[
    V_{t}=w+\int_{0}^{t} \beta_{u}
    d 
    S_{u}-\int_{0}^{t} 
    c_{u} du.
\]
Here, $U_{1}$ is a (state-dependent) utility function and
$U_{2}(t,\cdot)$ is a utility function for each $t$.
The dual problem can now be defined as follows:
\[
    v_{_{\Gamma}}(y)=\inf_{\xi\in\Gamma}
    \bbe\left[\widetilde{U}_{1}\left(y\xi_{_{T}}
    \right)
    +\int_{0}^{T}\widetilde{U}_{2}(s,y\xi_{s}
    ds\right],
\]
over a suitable class of supermartingales $\Gamma$. For instance,
if the support of $\nu$ is $[-1,\infty)$, then
$\Gamma$ can be all supermartingales $\xi$ such that
$0\leq{}\xi_{0}\leq{}1$ and $\{\xi_{t}
S_{t}\}_{t\leq{}T}$
is a supermartingale. The dual Theorem \ref{MainThrm1} can be
extended for this problem.
However, the existence of a wealth-consumption
strategy pair $(\beta,c)$ that attains the potential final
wealth induced by the optimal dual solution
(as in Section \ref{RplcbltySect}) requires
further work. We hope to address this problem
in a future publication.

\appendix

\section{Convex classes of exponential supermartingales}
 \label{CovClssSect}
 \setcounter{equation}{0}
The goal of this part is to establish the theoretical
foundations behind Theorem \ref{ExistenceDual}.
We begin by recalling an important
optional decomposition theorem due to F\"ollmer and Kramkov \cite{Follmer:1997}.
Given a family of supermartingales
$\calS$ satisfying suitable conditions, the result characterizes the
nonnegative exponential local supermartingales
\(
    \xi:=\xi_{0}\calE(X-A),
\)
where $X\in\calS$ and $A\in\calV^{+}$,
in terms of the so-called \emph{upper variation
process} for $\calS$.
Concretely, let $\calP(\calS)$ be the class of
probability measures $\bbq\sim\bbp$ for which
there is an increasing predictable process
$\left\{A_{t}\right\}_{t\geq{}0}$
(depending on $\bbq$ and $\calS$) such that
\(
    \{X_{t}-A_{t}\}_{t\geq{}0}
\)
is a local supermartingale under $\bbq$,
for all $X\in\calS$. The \emph{smallest}\footnote{That is, if $A$ satisfies such a property
then $A-A^{\calS}(\bbq)$ is increasing.}
of such processes $A$
is denoted by $A^{\calS}(\bbq)$ and is called the
upper variation process for $\calS$ corresponding
to $\bbq$.
For easy reference, we state F\"ollmer and Kramkov's
result (see \cite{Follmer:1997} for a proof).
\begin{thrm}
\label{ThrmFollmerKramkov}
Let $\calS$ be a family of semimartingales
 that are null at zero, and that are locally bounded from below. Assume
 that $0\in\calS$, and that the following
 conditions {hold:}
       \begin{enumerate}
       \item[(i)] $\calS$ is predictably convex;

 \smallskip
        \item[(ii)] $\calS$ is closed under the \'Emery distance;

 \smallskip
        \item[(iii)] $\calP(\calS)\neq \emptyset$.
        \end{enumerate}
        Then, the following two statements are equivalent for a
        nonnegative process $\xi$:
        \begin{enumerate}
            \item $\xi$ is of the form
            $\xi=\xi_{0}\calE\left(X-A\right)$,
            for some $X\in\calS$ and an increasing process
            $A\in\calV^{+}$;

        \smallskip
            \item $\frac{\xi}{\calE\left(A^{\calS}(\bbq)\right)}$
            is a supermartingale under $\bbq$ for each
            $\bbq\in\calP(\calS)$.
        \end{enumerate}
\end{thrm}

The next result is a direct consequence of the
previous representation. Recall that
a sequence of processes $\{\xi^{n}\}_{n\geq{}1}$ is said
to be ``Fatou convergent on $\pi$" to a process
$\xi$ if $\{\xi^{n}\}_{n\geq{}1}$  is uniformly bounded from below
and it holds that
   \begin{equation}\label{FatouConvEq}
        \xi_{t}=\limsup_{s\downarrow{}t\;:\;
        s\in\pi}\limsup_{n\rightarrow\infty} \xi^{n}_{s}=
        \liminf_{s\downarrow{}t\;:\;
        s\in\pi}\liminf_{n\rightarrow\infty} \xi^{n}_{s},
    \end{equation}
almost surely for all $t\geq{}0$.
\begin{prop}\label{ClosednessClasses}
    If $\calS$ is a class of semimartingales satisfying the
    conditions in Theorem \ref{ThrmFollmerKramkov}, then
    \begin{equation}\label{ExpSmmartClass}
        \Gamma^{0}(\calS):= \{ \xi:=\xi_{0}\calE(X- A):
        X\in\calS,\, A\text{ increasing, and }\xi\geq{}0\},
    \end{equation}
    is convex and closed under Fatou convergence on
    any fix dense countable set $\pi$ of $\bbr_{+}$;
    that is, if
    $\{\xi^{n}\}_{n\geq 1}$ is a sequence in
    $\Gamma^{0}(\calS)$ that is Fatou convergent on $\pi$
    to a process $\xi$, then $\xi\in\Gamma^{0}(\calS)$.
\end{prop}

\begin{proof}
The convexity of $\Gamma^{0}(\calS)$ is a direct consequence of
Theorem \ref{ThrmFollmerKramkov}, since the convex combination of
supermartingales remains a supermartingale. Let us prove the
closure property.
Fix a $\bbq\in\calP(\calS)$ and denote $C_{t}:=\calE\left(
A^{\calS}(\bbq)\right)$. Notice that $C_{t}>0$ because
$A^{\calS}(\bbq)_{t}$ is increasing and hence, its jumps are
nonnegative. Since $\xi^{n}\in\Gamma^{0}(\calS)$,
$\{C_{t}^{-1}\xi^{n}_{t}\}_{t\geq{}0}$ is a supermartingale under
$\bbq$. Then, for $0<s'<t'$,
\[
    \bbe^{\bbq}\left[C^{-1}_{t'}\xi^{n}_{t'}|\calF_{s'}\right]
    \leq{}C^{-1}_{s'}\xi^{n}_{s'}.
\]
By Fatou's Lemma and the right-continuity of process $C$,
\[
    \bbe^{\bbq}\left[
    C^{-1}_{t}\xi_{t}|\calF_{s'}\right]=
    \bbe^{\bbq}\left[
    \liminf_{t'\downarrow{}t:t'\in\pi}
    \liminf_{n\rightarrow\infty}
    C^{-1}_{t'}\xi^{n}_{t'}|\calF_{s'}\right]
    \leq{}C^{-1}_{s'}\xi^{n}_{s'}.
\]
Finally, using the right-continuity of the filtration,
\[
    \bbe^{\bbq}\left[
    C^{-1}_{t}\xi_{t}|\calF_{s}\right]\leq
    \liminf_{s'\downarrow{}s:s'\in\pi}\liminf_{n\rightarrow\infty}
    C^{-1}_{s'}\xi^{n}_{s'}=
    C_{s}^{-1}\xi_{s},
\]
where $0\leq{}s<t$. Since $\bbq$ is arbitrary, the
characterization of Theorem \ref{ThrmFollmerKramkov} implies that
$\xi\in\Gamma^{0}(\calS)$.
\end{proof}

The most technical condition in Theorem \ref{ThrmFollmerKramkov}
is the closure property under \'Emery distance.
 The following result is useful to deal with this condition.
It shows that the class of integrals with respect to a Poisson random
measure is closed with respect to \'Emery distance, thus extending
the analog property for integrals with respect to a fixed
semimartingale due to M\'emin \cite{Memin}.

\begin{thrm}\label{ClosenessResult}
    {Let $\Theta$ be a closed convex subset of $\bbr^{2}$
    containing the origin.}
    Let $\Pi$ be the set of  all predictable processes
    $(F,G)$,
    $F\in G_{loc}(N)$ and  $G\in L^{2}_{loc}(W)$,
    such that $F(t,\cdot)=G(t)=0$, for all $t\geq{}T$,
    and $(F(\omega,t,z),G(\omega,t))\in\Theta$,
    for $\bbp\times dt\times\nu(dz)$-a.e.
        $(\omega,t,z)\in\Omega\times\bbr_{+}\times\bbr_{0}$.
    Then, the class
    \begin{equation}\label{BscClass2}
    \calS:=\left\{X_{t}:=
    \int_{0}^{t}G(s)d W_{s}
    +\int_{0}^{t}\int_{\bbr_{0}}F(s,z)\widetilde{N}(ds,dz):
    (F,G)\in\Pi\right\}
    \end{equation}
    is closed under convergence with respect to
    \'Emery's topology.
\end{thrm}

\begin{proof}
Consider a sequence of semimartingales
\[
    X^{n}(t):=\int_{0}^{t} G^{n}(s) d W_{s}+\int_{0}^{t}\int_{\bbr} F^{n}(s,z)\widetilde{N}(ds,dz), \quad n\geq{}1,
\]
in the class $\calS$. Let $X$ be a semimartingale such that
$X^{n}\rightarrow X$ under \'Emery topology. To prove the result,
we will borrow some results in \cite{Memin}.

For some $\bbq\sim\bbp$, we denote $\calM^{2}(\bbq)$
to be the Banach space of all
$\bbq$-square integrable martingales on $[0,T]$, endowed with the
norm $\|M\|_{\calM^{2}(\bbq)} :=
        \left(\bbe^{\bbq}\left<M,M\right>_{T}\right)^{1/2}
        =\left(\bbe^{\bbq}\left[M,M\right]_{T}\right)^{1/2}$,
and $\calA(\bbq)$ to be the Banach space of all predictable processes on
$[0,T]$ that have $\bbq$-integrable total variations, endowed with
the norm
$\|A\|_{\calA(\bbq)} :=\bbe^{\bbq}{\rm Var}(A)$.
Below, $\calA_{loc}^{+}(\bbq)$ stands for the localized class
of increasing process in $\calA(\bbq)$.
By Theorem II.3 in \cite{Memin}, one can extract a
subsequence from $\{X^{n}\}$, still denote it by $\{X^{n}\}$,
{for} which one can construct a probability measure $\bbq$,
defined on $\calF_{T}$ and equivalent to $\bbp_{_{T}}$ (the
restriction  of $\bbp$ on $\calF_{T}$), such that the following
assertions hold:
\begin{enumerate}
    \item[(i)]
    $\xi:=\frac{d\bbq}{d\bbp_{_{T}}}$
    is bounded by a constant;
    \smallskip
    \item[(ii)] $X^{n}_{t}=M^{n}_{t}+A^{n}_{t}$, $t\leq{}T$, for
    Cauchy sequences $\{M^{n}\}_{n\geq{}1}$
    and $\{A^{n}\}_{n\geq{}1}$ in
    $\calM^2(\bbq)$ and $\calA(\bbq)$, respectively.
\end{enumerate}

Let us extend $M^{n}$ and $A^{n}$ to $[0,\infty)$ by
 setting $M^n_t=M^n_{t\wedge T}$ and $A^n=A^n_{t\wedge T}$ for
 all $t\ge 0$.
Also, we extend $\bbq$ for $A\in\calF$ by setting $\bbq(A):=\int_{A} \xi
d\bbp$, so that $\bbq\sim\bbp$ (on $\calF$). In
that case, it can be proved that
$
    \calA^{+}_{loc}(\bbp)=\calA^{+}_{loc}(\bbq).
$
This follows essentially from Proposition III.3.5 in
\cite{Shiryaev:2003} and Doob's Theorem.
Now, let
\(
    \xi_{t}:=\frac{d\bbq|_{\calF_{t}}}{d\bbp|_{\calF_{t}}}=
    \bbe\left[\xi|\calF_{t}\right],
\)
denote the density process. Since $\xi$ is bounded, both
$\{\xi_{t}\}_{t}$ and
$\{|\Delta \xi_{t}|\}_{t}$ are bounded. By Lemma III.3.14 and Theorem
III.3.11 in \cite{Shiryaev:2003}, the $\bbp-$quadratic covariation
$[X^{n},\xi]$ has $\bbp-$locally integrable variation and the unique
canonical decomposition $M^{n}+A^{n}$ of $X^{n}$ relative to
 $\bbq$ is given by
\[
    M^{n}= X^{n}-\int_{0}^{t} \frac{1}{\xi_{s-}}d\left<X^{n},\xi\right>_{s},\quad
    A^{n}=\int_{0}^{t} \frac{1}{\xi_{s-}}d\left<X^{n},\xi\right>_{s}.
\]
Also, the $\bbp$-quadratic variation of the continuous part
$X^{n,c}$ of $X^{n}$ (relative to $\bbp$), given by
\(
    \left<X^{n,c},X^{n,c}\right>_{\cdot}=\int_{0}^{\cdot}
    \left(G^{n}(s)\right)^{2}ds,
\)
is also a version of  the $\bbq$-quadratic variation of the
continuous part of $X^{n}$ (relative to $\bbq$).
By the representation
theorem for local martingales relative to $Z$ (see e.g. Theorem
III.4.34 in \cite{Shiryaev:2003} or Theorem 2.1 in
\cite{Kunita:2004}), $\xi$ has the representation
\[
    \xi_{t}=1+\int_{0}^{t} \xi_{s^{-}}E(s) dW_{s}+\int_{0}^{t}\int_{\bbr}\xi_{s^{-}}D(s,z)\widetilde{N}(ds,dz),
\]
for predictable $D$ and $E$ necessarily satisfying that
$D>-1$,
\[
    \bbe\int_{0}^{T}\int_{\bbr} D^{2}(s,z)\xi_{s}^{2}
    \nu(dz)ds<\infty,
    \quad {\rm and}\quad
    \bbe\int_{0}^{T} E^{2}(s)\xi_{s}^{2}
    ds<\infty.
\]
Then,
\begin{align*}
    \left<X^{n},\xi\right>_{t}&=
    \int_{0}^{t} G^{n}(s) E(s) \xi_{s^{-}}ds+
    \int_{0}^{t}\int_{\bbr} F^{n}(s,z)D(s,z)\xi_{s^{-}} \nu(dz)ds,\\
    A^{n}_{t}&=\int_{0}^{t}\int_{\bbr} F^{n}(s,z)D(s,z)
    \nu(dz)ds
    +\int_{0}^{t} G^{n}(s) E(s) ds.
\end{align*}
We conclude that $\Delta M^{n}_{t}=\Delta X^{n}_{t}=F^{n}(t,\Delta
Z_{t})$. Hence, $\Delta M^{n}=\Delta \widetilde{M}^{n}$, where
$\widetilde{M}^{n}$ is the purely discontinuous local
martingale (relative to $\bbq$) defined by
\[
    \widetilde{M}^{n}_{t}:=
    \int_{0}^{t}\int_{\bbr} F^{n}(s,z)\left(N(ds,dz)-\nu^{\bbq}(ds,dz)\right),
\]
where $\nu^{\bbq}(ds,dz):=Y(s,z))ds\nu(dz)$ is the compensator of
$N$ relative to $\bbq$ (see Theorem III.3.17 in
\cite{Shiryaev:2003}). It can be shown that $Y=1+D$. Notice that
$\widetilde{M}^{n}$ is well-defined since
$\calA^{+}_{loc}(\bbp)=\calA^{+}_{loc}(\bbq)$ and
the Definition III.1.27 in \cite{Shiryaev:2003}. Then, the purely
discontinuous part of the local martingale $M^{n}$ (relative to
$\bbq$) is given by $\widetilde{M}^{n}$ (see I.4.19 in
\cite{Shiryaev:2003}), and since $M^{n}\in\calM^{2}(\bbq)$,
\begin{align*}
    \bbe^{\bbq}\left[M^{n},M^{n}\right]_{T}
    &
    \bbe^{\bbq}\int_{0}^{T} \left(F^{n}(s,z)\right)^{2}
    Y(s,z) \nu(dz)ds
    +\bbe^{\bbq}\int_{0}^{T}
    \left(G^{n}(s)\right)^{2}ds<\infty.
\end{align*}
Similarly, since $\{M^{n}\}_{n\geq{}1}$ is a Cauchy sequences
under the norm $\bbe^{\bbq}\left[M,M\right]_{T}$,
\begin{align*}
    \bbe^{\bbq}\left[M^{n}-M^{m},M^{n}-M^{m}\right]_{T}
    &=
    \bbe^{\bbq}\int_{0}^{T} \left(F^{n}(s,z)-F^{m}(s,z)\right)^{2} Y(s,z) \nu(dz)ds\\
    \quad&+\bbe^{\bbq}\int_{0}^{T}
    \left(G^{n}(s)-G^{m}(s)\right)^{2}ds\rightarrow{}0,
\end{align*}
as $n,m\rightarrow\infty$. Using the notation
$\widetilde{\Omega}:=\Omega\times\bbr_{+}\times\bbr$ and
$\widetilde{\calP}:=\calP\times\calB(\bbr)$, where $\calP$ is the
predictable $\sigma-$ field, we conclude that
$\{F^{n}\}_{n\geq{}1}$ is a Cauchy sequence in the Banach space
\[
    \bbh_{d}:=\bbl^{2}\left(\widetilde\Omega,
    \widetilde\calP,Y \,d\bbq\,d\nu\,dt\right)\cap
    \bbl^{1}\left(\widetilde\Omega,
    \widetilde\calP,
    \left|D\right|\,d\bbq\,d\nu\,dt\right),
\]
and thus, there is $F\in\bbh_{d}$ such that $F^{n}\rightarrow F$,
as $n\rightarrow\infty$. Similarly, there exists a $G$ in the
Banach space
\[
    \bbh_{c}:=\bbl^{2}\left(\Omega\times\bbr_{+},
    \calP,d\bbq\,dt\right)\cap
    \bbl^{1}\left(\Omega\times\bbr_{+},
    \calP,
    \left|E\right|\,d\bbq\,d\nu\,dt\right),
\]
such that $G^{n}\rightarrow G$, as $n\rightarrow\infty$. In
particular, $(F,G)$ satisfies condition (iv) since
$Y=1+D$ is strictly positive, and each $(F^{n},G^{n})$ satisfies
(iv). Also, $F\in G_{loc}(N)$ relative to $\bbq$
in light of $\calA^{+}_{loc}(\bbp)=\calA^{+}_{loc}(\bbq)$. Similarly,
$\int_{0}^{\cdot} G^{2}(s) ds$ belongs to $\calA^{+}_{loc}(\bbq)$,
and hence, belongs to $\calA_{loc}^{+}(\bbp)$. It follows that the
process
\[
    \widetilde{X}:=\int_{0}^{t} G(s) d W_{s}+\int_{0}^{t}\int_{\bbr} F(s,z)\widetilde{N}(ds,dz), \quad n\geq{}1,
\]
is a well-defined local martingale relative to $\bbp$. Applying
Girsanov's Theorem to $\widetilde{X}$ relative to $\bbq$ and
following the same argument as above, the purely discontinuous
local martingale and bounded variation parts of $\widetilde{X}$
are respectively
\begin{align*}
    M^{d}_{t}&=
    \int_{0}^{t}\int_{\bbr} F(s,z)\left(N(ds,dz)-\nu^{\bbq}(ds,dz)\right),
    \\
    A_{t}&=\int_{0}^{t}\int_{\bbr} F(s,z)D(s,z)
    \nu(dz)ds
    +\int_{0}^{t} G(s) E(s)
    d s.
\end{align*}
The continuous part of $\widetilde{X}$ has quadratic variation
$\int_{0}^{\cdot} G^{2}(s)ds$. We conclude that
$\widetilde{X}\in\calM^{2}(\bbq)\oplus\calA(\bbq)$ and
$X^{n}\rightarrow \widetilde{X}$ on
$\calM^{2}(\bbq)\oplus\calA(\bbq)$. Then, $X^{n}$ converges under
\'Emery's topology to $\widetilde{X}$ and hence,
$X=\widetilde{X}$.
\end{proof}

%

\end{document}